\documentclass{article}

\usepackage[letterpaper,margin=1.00in]{geometry}
\usepackage{amsmath, amsthm, amsfonts}
\usepackage{bbm}

\usepackage{ifthen}
\usepackage{tikz}
\usetikzlibrary{positioning,decorations.pathreplacing}

\usepackage{cite}
\usepackage{appendix}
\usepackage{graphicx}
\usepackage{color}
\usepackage{color}
\usepackage{algorithm}
\usepackage{algorithm}
\usepackage[noend]{algpseudocode}
\usepackage{epstopdf}
\usepackage{wrapfig}
\usepackage{paralist}
\usepackage{wasysym}
\usepackage[textsize=tiny]{todonotes}

\usepackage{framed}
\usepackage[framemethod=tikz]{mdframed}
\usepackage[bottom]{footmisc}
\usepackage{enumitem}
\setitemize{noitemsep,topsep=3pt,parsep=3pt,partopsep=3pt}
\usepackage[font=small]{caption}
\usepackage{xspace}

\newtheorem{theorem}{Theorem}[section]
\newtheorem{lemma}[theorem]{Lemma}
\newtheorem{meta-theorem}[theorem]{Meta-Theorem}
\newtheorem{claim}[theorem]{Claim}
\newtheorem{remark}[theorem]{Remark}
\newtheorem{corollary}[theorem]{Corollary}

\newtheorem{definition}[theorem]{Definition}

\definecolor{darkgreen}{rgb}{0,0.5,0}
\usepackage{hyperref}
\hypersetup{
	unicode=false,          
	colorlinks=true,        
	linkcolor=red,          
	citecolor=darkgreen,        
	filecolor=magenta,      
	urlcolor=cyan           
}
\usepackage[capitalize, nameinlink]{cleveref}
\Crefname{remark}{Remark}{Remarks}
\Crefname{observation}{Observation}{Observations}
\Crefname{claim}{Claim}{Claims}
\algnewcommand\algorithmicswitch{\textbf{switch}}
\algnewcommand\algorithmiccase{\textbf{case}}

\algdef{SE}[SWITCH]{Switch}{EndSwitch}[1]{\algorithmicswitch\ #1\ \algorithmicdo}{\algorithmicend\ \algorithmicswitch}%
\algdef{SE}[CASE]{Case}{EndCase}[1]{\algorithmiccase\ #1}{\algorithmicend\ \algorithmiccase}%
\algtext*{EndSwitch}%
\algtext*{EndCase}%

\newcommand{\eps}{\varepsilon}
\renewcommand{\epsilon}{\varepsilon}

\newcommand{\poly}{\operatorname{\text{{\rm poly}}}}

\DeclareMathOperator{\E}{\mathbb{E}}

\newcommand{\Abs}[1]{\left\lvert#1\right\rvert}

\begin{document}
\date{}
\title{Average Awake Complexity of MIS and Matching}

\author{Mohsen Ghaffari \\
	MIT \\
	ghaffari@mit.edu
	\and
	Julian Portmann\\
	ETH Zurich \\
	pjulian@inf.ethz.ch
}

\maketitle

\begin{abstract}
	Chatterjee, Gmyr, and Pandurangan [PODC 2020] recently introduced the notion of \textit{awake complexity} for distributed algorithms, which measures the number of rounds in which a node is \textit{awake}.
	In the other rounds, the node is \textit{sleeping} and performs no computation or communication.
	Measuring the number of awake rounds can be of significance in many settings of distributed computing, e.g., in sensor networks where energy consumption is of concern.

	\smallskip
	In that paper, Chatterjee et al. provide an elegant randomized algorithm for the Maximal Independent Set (MIS) problem that achieves an $O(1)$ \textit{node-averaged awake complexity}.
	That is, the average awake time among the nodes is $O(1)$ rounds.
	However, to achieve that, the algorithm sacrifices the more standard \textit{round complexity} measure from the well-known $O(\log n)$ bound of MIS, due to Luby [STOC'85], to $O(\log^{3.41} n)$ rounds.

	\smallskip
	Our first contribution is to present a simple randomized distributed MIS algorithm that, with high probability, has $O(1)$ node-averaged awake complexity and $O(\log n)$ worst-case round complexity.
	Our second, and more technical contribution, is to show algorithms with the same $O(1)$ node-averaged awake complexity and $O(\log n)$ worst-case round complexity for $(1+\varepsilon)$-approximation of maximum matching and $(2+\varepsilon)$-approximation of minimum vertex cover, where $\varepsilon$ denotes an arbitrary small positive constant.
	\smallskip
\end{abstract}

\section{Introduction}
Distributed algorithms for local graphs problems --- such as maximal independent set, maximal matching, and vertex and edge coloring --- have been studied extensively over the past four decades.
The primary focus in these studies has been on the time complexity of the algorithm, e.g., the number of rounds until all nodes terminate.
Of course, there are other resources besides \textit{time} whose expenditure is not captured accurately enough by this measure.
Recently, Chatterjee, Gmyr, and Pandurangan~\cite{ChaterjeeGP20} introduced \emph{awake complexity} as a better measure for some of the other resources, most importantly \textit{energy} (which is a critical resource, e.g., in sensor networks).
Below, we review the model and this measure, and after describing prior results, we outline our contributions.

\subsection{Model and Motivation}
The baseline setup is an abstraction of the network as a synchronous message passing model.
Concretely, we work with the $\mathsf{CONGEST}$ model~\cite{peleg00}.

\paragraph{The $\mathsf{CONGEST}$ Model} The network is abstracted as an $n$-node graph $G=(V, E)$ and per round each node can send one $O(\log n)$ bit message to each of its neighbors.
The variant where message sizes are unbounded is known as the $\mathsf{LOCAL}$ model~\cite{linial1987LOCAL, peleg00}.
Initially, nodes do not know the topology of the network, except for some estimates of global parameters, i.e., a polynomially-tight upper bound on $n$ and the maximum degree $\Delta$.
In the end, each node should know its own part of the output, e.g., in the maximal independent set problem, the node should know whether it is in the computed maximal independent set or not.

\paragraph{Sleeping/Awake Rounds, and Awake Complexity}
In their variant of the model, Chatterjee et al. allow each node to be awake or sleeping, per round.
A sleeping node receives no communication and performs no computation.
The awake complexity $A_v$ of a node $v$ is simply the number of rounds in which $v$ is awake, until its termination.

\paragraph{Motivation} Among others, one key motivation for studying the awake complexity is that it provides a sharper bound on the energy spent.
Notice that energy is one of the key resources in various distributed systems, especially in practical settings where algorithms for local graph problems might be deployed, e.g., sensor networks and ad hoc networks.
Since a sleeping node performs no computation or communication, it is reasonable to model the energy spent during a sleeping round as zero (i.e., negligible).
In this sense, the energy spent by one node $v$ can be modeled as proportional to its awake complexity $A_v$.

\paragraph{Node-Averaged Awake Complexity} When studying energy consumption, besides the maximum energy spent by each node --- which we will refer to as \textit{worst-case awake complexity}--- it is interesting (and arguably even more relevant) to understand the energy spent by the entire network.
Similar to above, the total energy spent in the network can be modeled as (proportional to) $\sum_{v\in V} A_v$, i.e., the summation over all nodes $v$ of the number rounds in which $v$ is awake.
Hence, to measure the total energy spent in the network, we can focus on the \textit{node-averaged awake complexity}, i.e., $\frac{1}{n} \sum_{v\in V} A_v.$ See \Cref{subsec:Node-Averaged} for more on node-averaged measures.

\subsection{State of the Art}
The maximal independent set (MIS) problem and the matching problems---maximal matching and approximations of maximum matching---are among the central problems in distributed graph algorithms, and there has been extensive work on them since the 1980s.
We first review the best known bounds for the standard measure of round complexity, and then discuss known results regarding awake complexity.

\paragraph{State of the Art for Round Complexity --- Randomized Algorithms}
The classic randomized MIS algorithms, by Luby~\cite{luby86} and Alon, Babai, and Itai~\cite{alon86}, have round complexity $O(\log n)$, with high probability\footnote{As standard, we use the phrase \textit{with high probability} (w.h.p.) indicate that an event happens with probability at least $1-1/n^{c}$ for a desirably large constant $c\geq 2$}.
This also directly provides an $O(\log n)$ round algorithm for maximal matching, as maximal matching easily reduces to MIS on the line graph.
This itself implies a $2$-approximation of the maximum matching in $O(\log n)$ rounds.
Lotker, Patt-Shamir, and Pettie~\cite{lotkerMatchingImproved} showed how to improve the approximation to $(1+\epsilon)$ for maximum matching, for an arbitrarily small constant $\epsilon>0$, while keeping the same $O(\log n)$ round complexity\footnote{We do not discuss the extensive literature of approximation of maximum \emph{weighted} matching, as this paper targets only unweighted matchings}.
Finally, Bar-Yehuda, Censor-Hillel, Ghaffari, and Schwartzman~\cite{bar2017distributed} improved the round complexity of $(1+\epsilon)$-approximation to $O(\log \Delta/\log\log \Delta)$, where $\Delta$ denotes the maximum degree, which in the worst-case of $\Delta = \Theta(n)$ leads to a round complexity of $O(\log n/\log\log n)$.
Thus, to conclude, the best known round complexity of all these problems remains roughly at $O(\log n)$, with the exception of matching approximation which is marginally faster at $O(\log n/\log\log n)$ rounds.

\paragraph{State of the Art for Round Complexity --- Deterministic Algorithms} Deterministic algorithms for these problems are considerably slower.
We still review the state of the art for completeness.
For MIS, the best known deterministic algorithm for a long time remained at the $2^{O(\sqrt{\log n})}$ bound of Panconesi and Srinivasan~\cite{panconesi-srinivasan} but recently that was improved to $O(\log^7 n)$ by Rozho{\v{n}} and Ghaffari~\cite{rozhovn2019networkDecomposition} and further to $O(\log^5 n)$ by Ghaffari, Grunau, and Rozho{\v{n}}~\cite{ghaffari2021improvedNC}.
For maximal matching, while the same bounds also apply as MIS, somewhat faster algorithms have been known.
Hanckowiak, Panconesi, and Karonski~\cite{hanckowiak01} gave an $O(\log^4 n)$ round deterministic maximal matching algorithm, and the round complexity was improved to $O(\log^2 \Delta \log n)$ by Fischer~\cite{fischer2020improved}, and that algorithm can also provide a $(2+\epsilon)$-approximation of maximum matching in $O(\log^2 \Delta + \log^* n)$ rounds.

\paragraph{Improved Algorithms for Special Graphs} Ignoring the distinction between randomized and deterministic algorithms, the roughly $O(\log n)$ bound remains the best known round complexity, for arbitrary graphs.
However, there have been improvements in special cases, particularly when the maximum degree is not too large.

For graphs with very small maximum degree $\Delta$ --- concretely $\Delta \leq O(\log n)$--- a faster algorithm follows from the $O(\Delta+ \log^* n)$-round deterministic algorithm of Panconesi and Rizzi~\cite{panconesi2001some}.
For moderately small values of $\Delta$ ---concretely when $\log \Delta = o(\log n)$ and $\Delta\geq \poly(\log\log n)$---a faster algorithm randomized algorithm is known, with round complexity $O(\log \Delta) + \poly(\log\log n)$, which follows from the work of Ghaffari~\cite{gmis}, when combined with the recent results of (and as described in) Rozho{\v{n}} and Ghaffari ~\cite{rozhovn2019networkDecomposition}.
This builds on the shattering approach of Barenboim et al.\cite{barenboim_symmbreaking}.
There is one caveat: this algorithm requires large messages.
For the setting with $O(\log n)$-bit messages, improved bounds were provided by \cite{pai2017symmetry, ghaffari2019MIS, ghaffariP2019improved} and the best known round complexity is $O(\log \Delta \cdot \sqrt{\log\log n}) + 2^{O(\sqrt{\log\log n})}$.
For maximal matching, the best known complexity is $O(\log \Delta + \log^3\log n)$, by combining the work of ~\cite{gmis} and ~\cite{fischer2020improved}.
For matching approximation, the aforementioned work of Bar-Yehuda et al.\cite{bar2017distributed} achieves a complexity of $O(\log \Delta/\log\log \Delta)$ for $(1+\epsilon)$-approximation, which is the optimal round complexity for $\Delta = O(\sqrt{\log n})$, matching a lower bound of Kuhn et al.~\cite{lowerbound, kuhn16_jacm}.

There are many other algorithmic results for MIS and matching for other restricted graph families, e.g., graphs of bounded arboricity~\cite{barenboim2010sublogarithmic, barenboimelkin_book}, growth bounded graphs~\cite{schneider2008log}, etc.

\paragraph{Lower Bounds} An $\Omega(\log^* n)$ round lower bound is known due to the seminal work of Linial~\cite{linial1987LOCAL}, which holds even in the case of the network being a cycle.
This applies to both MIS and maximal matching, and it also extends to randomized algorithms~\cite{naor1991lower}.
A lower bound of $\Omega(\min\{\frac{\log \Delta}{\log\log \Delta}, \sqrt{\frac{\log n}{\log\log n}}\})$ was presented by Kuhn, Moscibroda, and Wattenhofer~\cite{lowerbound, kuhn16_jacm} which holds also for randomized algorithm, and applies to MIS, maximal matching, and even any constant approximation of maximum matching.

Finally, Balliu et al.\cite{balliu2019LB} recently proved that there is no randomized algorithm with round complexity $o(\Delta)+o(\log\log n/\log\log\log n)$ and there is no deterministic algorithm with round complexity $o(\Delta)+o(\log n/\log\log n)$, either for MIS or for maximal matching.
All the aforementioned lower bounds hold in the $\mathsf{LOCAL}$ model where arbitrary size messages are allowed.

\subsection{Beyond Worst-Case}
To summarize the above, the key takeaway is that the best known round complexity for arbitrary graphs remains at roughly $O(\log n)$ rounds~\cite{luby86, alon86}, modulo a $\log\log n$ factor in the case of maximum matching approximation.
It is also known that one cannot achieve a worst-case round complexity better than $\Omega(\min\{\sqrt{\frac{\log n}{\log\log n}}, \frac{\log \Delta}{\log\log \Delta}\})$\cite{kuhn16_jacm}, for either of MIS or matching (any constant approximation of maximum, thus including maximal matching) problems.
Even further, in a recent and independent work, Balliu, Ghaffari, Kuhn, and Olivetti~\cite{BGKO22} show that this $\Omega(\min\{\sqrt{\frac{\log n}{\log\log n}}, \frac{\log \Delta}{\log\log \Delta}\})$ lower bound holds even for the node-average round complexity of any algorithm for those problems (see \Cref{subsec:Node-Averaged} for some improved bounds on node-averaged round complexity in special graph families).

Having faced with this barrier for round complexity and even the node-averaged round complexity, one remaining hope is that perhaps the number of rounds in which nodes are awake can be much smaller.
That would for instance imply that the energy expenditure can be considerably lower.
Though, since a node performs no computation or communication during sleeping rounds, a priori, it is unclear if the awake complexity can be much lower.
Indeed, we are not aware of an algorithm with $O(\log n/\log\log n)$ worst-case awake complexity for MIS or any constant matching approximation.
But for node-averaged awake complexity, something better is known at least for MIS, as we review next.

\paragraph{State of the Art for Node-Averaged Awake Complexity}  Chatterjee et al.~\cite{ChaterjeeGP20} present a randomized MIS algorithm that achieves a node-averaged awake complexity of $O(1)$.
Also, the worst-case awake complexity of this algorithm is $O(\log n)$, that is, each node is awake at most $O(\log n)$ rounds.
But unfortunately, this algorithm ends up sacrificing the most standard measure, i.e., round complexity.
In particular, the round complexity of this algorithm is $O(\log^{3.41} n)$, with high probability, and even the node-averaged round complexity is $O(\log^{3.41} n)$.
This is considerably higher than the usual $O(\log n)$ bound discussed above~\cite{luby86, alon86}.
This leaves one natural question:
\medskip
\begin{center}
	\begin{minipage}{0.9\textwidth}
		\textit{\textbf{Question:} Can we achieve this ideal $O(1)$ node-averaged awake complexity while keeping round complexity $O(\log n)$?}
	\end{minipage}
\end{center}
Furthermore, the standard connection between MIS and (maximal) matching does not provide a node-averaged awake complexity.
Notice that if we apply an MIS algorithm on the line graph to solve the maximal matching, then the nodes of the MIS problem are the edges of the matching problem and the node-averaged awake complexity only would imply something about the average time that each edge is active.
This certainly does not imply anything useful for the node-averaged awake complexity of (maximal) matching.
As such, another natural question to ask is  this:
\begin{center}
	\begin{minipage}{0.9\textwidth}
		\textit{\textbf{Question:} Can we achieve the $O(1)$ node-averaged awake complexity also for matching, concretely for maximal matching, or any constant approximation of maximum matching? Ideally, we would also keep the round complexity $O(\log n)$.}
	\end{minipage}
\end{center}

\subsection{Our Contribution}
\paragraph{Contribution for MIS.}
Our first contribution answers the first question mentioned above about MIS.
Concretely, we prove that:
\begin{theorem}\label{thm:main}
	There is a randomized (Las Vegas) distributed MIS in the $\mathsf{CONGEST}$ model with the following performance, on any $n$-node network, w.h.p.:
	\begin{itemize}
		\item $O(1)$ node-averaged awake complexity, and
		\item $O(\log n)$ worst-case round complexity.
	\end{itemize}
\end{theorem}

This improves on the recent work of Chatterjee et al.~\cite{ChaterjeeGP20} who achieved $O(1)$ node-averaged awake complexity and $O(\log^{3.41} n)$ worst-case round complexity.
Our result effectively shows that we do not need to sacrifice round complexity, for node-averaged awake complexity.
As a side remark, we note that our algorithm is fairly simple; indeed the algorithm description and its analysis are considerably shorter than those of Chatterjee et al.~\cite{ChaterjeeGP20}.

\paragraph{Contribution for matching and vertex cover.} As discussed before, unlike the case of the worst-case complexity, an MIS algorithm with a good node-averaged (awake) complexity does not necessarily lead to a (maximal) matching algorithm with a good node-average (awake) complexity.
Our second and more technical contribution is to provide algorithms that compute a $(1+\epsilon)$-approximation of maximum matching and $(2+\eps)$-approximation of minimum vertex cover with $O(1)$ node-averaged awake complexity and  $O(\log n)$ with worst-case round complexity.

\begin{theorem}\label{thm:mainMM}
	There are randomized distributed algorithms in the $\mathsf{CONGEST}$ model that compute a $(1+\epsilon)$-approximation of maximum matching and $(2+\epsilon)$-approximation of minimum vertex cover, with the following performance, on any $n$-node network, w.h.p.:
	\begin{itemize}
		\item $O(1)$ node-averaged awake complexity, and
		\item $O(\log n)$ worst-case round complexity.
	\end{itemize}
\end{theorem}

We remark that in the above statement, the awake and round complexity bounds hold with high probability, i.e., probability  of at least $1-1/n^{c}$ for a desirably large constant $c\geq 2$.
The approximation bounds hold in expectation.
This is a standard phrasing of approximation when discussing randomized approximation algorithms.
In fact, our approximation has a stronger concentration around this expectation: the probability of having a worse approximation decays roughly exponentially in the maximum matching size (such a behavior is common in some randomized distributed/parallel algorithms for matching approximation, see e.g., \cite{bar2017distributed, ghaffari2018improved}).
We show our result in two versions: the easier version of expected approximation factor, and the stronger version of the approximation factor holding with high probability, for which we make the assumption that the maximum matching size is at least $\poly(\log n)$.
We defer tightening this assumption to the full version.

\subsection{Other Related Work on Node-Averaged Round Complexity}
\label{subsec:Node-Averaged}
It is worth noting that studying node-averaged awake complexity can be viewed as following the direction initiated by Feuilloley~\cite{feuilloley2020long}, who studied node-averaged round complexity.
Among other results, Feuilloley showed that Linial's $\Omega(\log^* n)$ round lower bound for $3$-coloring a cycle holds also for node-averaged round complexity.
He also exhibited some other sparse graphs where the node-averaged round complexity of $3$-coloring is $O(1)$, while the worst-case round complexity has to be $\Omega(\log^* n)$.
Barenboim and Tzur~\cite{barenboim2019distributed} provide further improvements in this direction, on various problems.
For the case of MIS, they show a deterministic algorithm with node-averaged round complexity of $O(a+\log^* n)$, where $a$ denotes the arboricity of the graph.
For the closely related problem of $\Delta+1$ coloring, the classic random algorithm of Johansson~\cite{johansson1999simple} achieves a node-averaged round complexity of $O(1)$, and worst-cast round complexity of $O(\log n)$.
The former is simply because, per round, each remaining node gets colored (and removed from the problem) with a constant probability.
As mentioned before, for MIS and matching, it was shown in concurrent and independent work by Alkida et al.~\cite{BGKO22} that the $\Omega(\min\{\sqrt{\frac{\log n}{\log\log n}}, \frac{\log \Delta}{\log\log \Delta}\})$ lower bound of Kuhn, Moscibroda, and Wattenhofer on the worst-case complexity can be adapted to hold even for the node-averaged round complexity.

\section{Preliminaries}
We make use of some basic concentration bounds:
\begin{theorem}[Hoeffding's Bound]
	\label[theorem]{thm:hoeffding}
	Let $X_1$, \dots, $X_N$ be indepndent random variables, each taking values in the range $[a_i, b_i]$.
	Let $\bar{X}=\sum_{i=1}^{N} X_i$.
	Then, for any $t\geq 0$, we have $$\Pr\big(|\bar{X} - \mathbb{E}[\bar{X}]| \geq t\big) \leq 2 exp(-\frac{2N^2 t^2}{\sum_{i=1}^{N} (b_i- a_i)^2}).$$
\end{theorem}

\begin{theorem}[Chernoff tight version]
	Let $X_1, \dots, X_n$ be independent random binary variables and let $a_1, \dots, a_n$ be coefficients in $[0,1]$.
	Let $X = \sum_{i=1}^n a_i X_i$.
	Then for any $\mu \geq \E [X]$ and $\delta > 0$
	\begin{align*}
		\Pr [X > (1 + \delta) \mu] \leq \left( \frac{e^\delta}{(1 + \delta)^{1 + \delta}} \right)^\mu.
	\end{align*}
	\label[theorem]{thm:ChernoffTight}
\end{theorem}

Besides the above, we need to make use of an extension of Chernoff/Hoeffding bound for cases where the random variables have small dependencies among each other.
This is formalized below.

\begin{theorem}[Chernoff Bound's Extension to Bounded-Dependecy Cases~\cite{pemmaraju2001equitable}]
	\label[theorem]{thm:Chernoff-Extension}
	Let $X_1$, \dots, $X_N$ be binary random variables, and let $a_1, \dots, a_n$ be coefficients in $[0,1]$.
	Suppose that each $X_i$ is mutually independent of all other variables except for at most $\Delta'$ of them.
	Let $\bar{X}=\sum_{i=1}^{N} a_i \cdot X_i$, and $\mu=\mathbb{E}[\bar{X}].$
	Then, for any $\delta \geq 0$, we have $$\Pr\big(|\bar{X} - \mathbb{E}[\bar{X}]| \geq \delta \mathbb{E}[\bar{X}]\big) \leq 8(\Delta'+1) \cdot exp(-\frac{\mathbb{E}[\bar{X}] \delta}{3(\Delta'+1)}).$$
\end{theorem}
Conveniently, setting $\Delta'=1$ gives us a way to use the standard Chernoff bound, modulo constant factors, for independent random variables.

\section{The MIS Algorithm}
\label{sec:MIS}
Here, we describe the MIS algorithm that proves \Cref{thm:main}, i.e., achieving $O(1)$ node-averaged awake complexity and $O(\log n)$ round complexity.

\paragraph{Algorithm Outline} The algorithm has three parts: During each part, we add some more nodes to the set $S$, which is initially empty and is always kept an \textit{independent set}, in a way that at the end $S$ is a \textit{maximal} independent set.
The first part brings down the maximum degree to $\poly(\log n)$.
The second part is the technical core of the new algorithm; it runs in $O(\log\log n)$ iterations each of which takes $O((\log\log n)^3)$ rounds, and per iteration, it removes a constant fraction of the remaining nodes.
This is also done in a way that the node-averaged awake complexity remains $O(1)$.
The third part is a simple clean-up step that deals with all the remaining nodes.
Next, we describe these parts, one by one.


\subsection{Part I}
\paragraph{Algorithm, Part I} For the first part, we use a partial run of the randomized greedy MIS algorithm~\cite{blelloch2012greedy, fischerNoever2018MIS}.
Recall that in the randomized greedy MIS algorithm, each node $v$ picks a random number $r_v$ in $I=\{1, 2, \dots, n^{10}\}$ and per round, all nodes that have a strictly smaller random number than their neighbors are added to the independent set $S$.
Then, they are removed from the graph along with their neighbors.
The key difference with Luby's algorithm~\cite{luby86} is that we do not use new random numbers for each round.

We run this algorithm only on the first $p=1/\log n$ portion of the random numbers.
Concretely, we run the algorithm only on nodes for which  the random number is in $I_p=\{1, 2, \dots, n^{10}/\log n\}$.
During this time, the rest of the nodes sleep.
In the end, using one additional round, we remove from the graph all of these nodes for which their random number is in $I_p$, as well as all nodes in $V\setminus I_p$ that are adjacent to the computed independent set $S$.

\paragraph{Analysis, Part I}
We now show that this first part of the algorithm reduces the maximum degree of the remaining graph to $\poly(\log n)$, and has node averaged awake complexity $O(1)$.
\begin{lemma}\label[lemma]{lem:part1}
	With high probability, we have the following properties about Part I of the algorithm, described above:
	\begin{enumerate}
		\item[(A)] this algorithm runs in $O(\log n)$ rounds,
		\item[(B)] this algorithm contributes $O(1)$ to the node-averaged awake complexity $\frac{1}{n} \sum_{v\in V} A_v$, and
		\item[(C)] in the subgraph induced by the remaining nodes, each node has degree at most $O(\log^2 n)$.
	\end{enumerate}
\end{lemma}
\begin{proof}
	By \cite{fischerNoever2018MIS}, the randomized greedy algorithm terminates in $O(\log n)$ rounds.
	Notice that the nodes with random numbers in $I_p$ must terminate by then, and they are not impacted by other nodes whose random number is larger than this range.
	Hence, even running the algorithm on only nodes with numbers in $I_p$ terminates in $O(\log n)$ rounds.

	Each node is awake during this time with probability $|I_p|/|I|= 1/\log n$.
	Thus, by Chernoff bound, at most $O(n/\log n)$ nodes are awake during the run of the algorithm.
	Considering that the run takes $O(\log n)$ rounds, this contributed at most $O(n/\log n) \cdot O(\log n) = O(n)$ to the summation $\sum_{v} A_v$ of the awake complexities of all nodes.
	In other words, this is a contribution of additive $O(1)$ to the node-averaged awake complexity $\frac{1}{n} \sum_{v} A_v$.

	Finally, by \cite[Lemma 4.3]{blelloch2012greedy}, we get the property that in the subgraph induced by the remaining nodes, each node has degree at most $O(\log^2 n)$, with high probability.
	The argument is simple, so let us repeat it here, to keep the article self-contained.
	The independent set computed during the algorithm is identical to the one that we would obtain if we process the nodes sequentially, one by one, according to increasing random numbers, each time putting the next node in the independent set $S$ and then removing all of its neighbors.
	Now consider an arbitrary node $v$.
	In each step of processing one node, this node is a random one among all remaining nodes.
	Hence, if $v$ has degree at least $d$, the probability that the next node is one of the neighbors of $v$ is at least $d/n$.
	In that case, $v$ would be removed from the graph.
	Hence, for one step, the probability that node $v$ remains and has a degree of at least $d$ is at most $(1-d/n)$.
	Over the $\Theta(pn)$ steps, we can conclude that the probability that $v$ remains is at most $(1-d/n)^{pn}$.
	Setting $d=10\log n/p = \Theta(\log^2 n)$, we get that the probability of a node $v$ remaining with degree above  $\Theta(\log^2 n)$ is at most $1/n^2$.
	A union bound over all vertices shows that, with probability at least $1-1/n$, no node remains with a degree of at least $d=\Theta(\log^2 n)$.
\end{proof}

\subsection{Part II}
Part II is the core novelty of our algorithm.
At the beginning of this part, we are left with a graph that has maximum degree at most $d=O(\log^2 n)$, with high probability.
For Part II, we have an algorithm with worst-case round complexity of $O((\log\log n)^4)$.
This consists of $O(\log\log n)$ iterations of repeating the same $O((\log\log n)^3)$-round procedure, each time on the remaining nodes.
The procedure will be such that, per iteration, the node-average awake complexity is $O(1)$, where the averaging is among the nodes that run this iteration.
Moreover, the number of remaining nodes for each iteration $i$ will be exponentially decreasing as a function of the iteration number $i$.
Hence, overall, the node-averaged awake complexity of Part II is $O(1)$.

\paragraph{Algorithm, Part II} We repeat the following procedure for $O(\log\log n)$ identical \textit{iterations}, each time on the remaining nodes.
The algorithm for each iteration has $\log d = O(\log\log n)$ \textit{phases}, where in phase $i \in [1, \log d]$, we have $C(\log d - i)^2$ \textit{rounds}: per round of the $i^{th}$ phase, we mark each node with probability $2^i/d$.
Then, any marked node that has no marked neighbor gets added to the independent set.
Here, $C$ is a sufficiently large constant, which we will discuss later.
At the end of the iteration, we have one extra round where each isolated node (a node with no remaining neighbor) joins the independent set $S$ and gets removed from the graph.

We now explain the times that nodes are awake or sleeping: At the start of each iteration, each node $v$ determines all the rounds in this iteration in which $v$ is marked.
This is done according to the marking probabilities described above.
Any node $v$ sleeps until the first time in this iteration in which $v$ is marked, and it stays awake from then until the iteration's end.
If there is no round in which $v$ is marked, then $v$ sleeps for the whole iteration, except for the very last round, where we removed isolated nodes.

Finally, we need to ensure that we remove neighbors of nodes that already joined the independent set $S$, even though the former might be sleeping.
For that, we split each round into two sub-rounds.
For any node $v$ that joined the independent set $S$ in a round $r$, in any round $r'\geq r+1$, in the first subround of round $r'$, node $v$ informs any neighbor $u$ that is marked (and is thus awake) in round $r'$ of the information that node $v$ is already in the computed independent set $S$.
In such a case, node $u$ terminates the entire algorithm and outputs that it is not in the independent set $S$.
The second subrounds are used for declaring whether a node $w$ is marked or not, and here only marked nodes $w$ that do not have neighbor already in the independent set $S$ (which $w$ would have learned in the first subround, if there was such a neighbor) declare to their neighbors that they are marked.
We note that this change of replacing each round with two subrounds does not impact the asymptotic bounds on round complexity or awake complexity.
For simplicity, we will ignore this $2$ factor and only refer to the number of rounds.

\paragraph{Analysis, Part II}
\begin{lemma}
	\label[lemma]{lem:part2}
	With high probability, we have the following properties about Part II of the algorithm, described above:
	\begin{itemize}
		\item[(A)] the algorithm runs in $O(\log^4 \log n)$ rounds.
		\item[(B)] the algorithm contributes $O(1)$ to the node-averaged awake complexity $\frac{1}{n}\sum_{v}A_v$, and
		\item[(C)] at most $O(n/\log n)$ nodes remain after the last iteration.
	\end{itemize}
\end{lemma}

\begin{proof}
	By the structure of the algorithm, it is clear that we have $O(\log\log n)$ iterations, each of which takes $\sum_{i=1}^{\log d} C(\log d - i)^2 = O((\log\log n)^3)$ rounds.
	Hence, the worst-case round complexity of this part of the algorithm is $O((\log\log n)^4)$.

	In each iteration, each node sleeps until the first round in which it is marked, and then it stays awake for the remainder of the iteration.
	The probability that a node $v$ wakes up in phase $i$ is at most $\frac{2^{i}}{d} \cdot C(\log d - i)^2$.
	In the case that $v$ wakes up in phase $i$, it stays awake for all the remaining rounds after that, in this iteration, which is at most $\big(\sum_{j=i}^{\log d} C(\log d - j)^2\big)$ many rounds.
	Hence, the expected number of awake rounds for each node that participates in an iteration is at most
	\begin{align*}
		\sum_{i=1}^{\log d} \frac{2^{i}}{d} \cdot C(\log d - i)^2 \cdot \big(\sum_{j=i}^{\log d} C(\log d - j)^2\big)
		\leq & \sum_{i=1}^{\log d} 2^{-(\log d-i)} \cdot O((\log d - i)^5) \\
		=    & \sum_{x=1}^{\log d} 2^{-x} \cdot O(x^5) = O(1).
	\end{align*}

	From the above, we know that the expected number of awake rounds for a node that participates in an iteration is $O(1)$.
	Let $V_k$ be the set of nodes that start iteration $k$.
	We will see soon, in the following paragraph, that $|V_k| \leq O(n/2^{k})$, with high probability.
	Moreover, the number of awake rounds for different nodes are independent random variables, and the maximum value of each is at most $O((\log\log n)^4)$.
	Hence, using Hoeffding's bound (cf. \Cref{thm:hoeffding}), we get that with high probability, for each $k\in [1, O(\log\log n)]$, the total summation $\sum_{v \in V_k} A_v$ of the number of awake rounds of the nodes that participate in the $k^{th}$ iteration is at most $O(n/2^{k})$.
	This means that, with high probability, the contribution to the node-averaged awake complexity $\frac{1}{n} \sum_{v} A_v$ during all these iterations is $\sum_{k=1}^{O(\log\log n)} O(1/2^{k}) = O(1)$.

	Finally, we show that each node $v$ is removed with at least a positive constant probability, per iteration.
	This will let us finish the proof of property (B) of the lemma, and also will allow us to prove property (C).

	We prove a more general statement: Fix an arbitrary node $v$ that starts phase $k$.
	We show that the degree of $v$ is upper bounded by an exponentially decreasing function throughout the phases of this iteration, with a certain probability.
	In fact, we will prove this fact not only for $v$ but for a certain radius of nodes around $v$, with a certain probability.
	As we go further and further, this radius will decrease, and eventually, it will reach radius zero, i.e., node $v$ itself.
	Moreover, the failure probability of our guarantee will also increase over time but will reach at most $1/2$ at the end.
	We next make this concrete.

	Let $\mathcal{E}_i$ be the event that, at the end of phase $i$, there is at least one node $u$ within distance $r_i=(\log d-i)$ of $v$ such that $u$ has degree at least $d/{2^{i}}$.
	We show using an induction on $i$ that $Pr[\mathcal{E}_i] \leq 2^{-2(\log d-i+1)^2}$.
	The claim holds trivially for the base case of $i=0$ as there all nodes have degree at most $d$ and thus $Pr[\mathcal{E}_0] = 0$.

	Consider a phase $i\in [1, \log d]$.
	From the induction hypothesis, we know that $Pr[\mathcal{E}_{i-1}] \leq 2^{-2(\log d-i+2)^2}$.
	Moreover, suppose that $\mathcal{E}_{i-1}$ did not happen.
	We already have an upper bound on the probability of $\mathcal{E}_{i-1}$, which we take into account using a union bound.

	Consider an arbitrary node $u$ within distance $r_i=(\log d-i)$ of $v$.
	As $\mathcal{E}_{i-1}$ did not happen, we know that at the start of phase $i$, node $u$ and any neighbor of it have a degree less than $d/2^{i-1}$.
	If $u$ has degree at least $d/2^{i}$, then per round of this phase, with at least a constant probability $c>0$, node $u$ or one of its neighbors is marked and has no marked neighbor.
	The reason is simple: consider one such round.
	Let us examine nodes in $N(u)\cup \{u\}$ one by one, each time revealing whether the node is marked or not, and we stop the moment we find one marked node $w\in N(u)\cup \{u\}$.
	Here, $N(u)$ denotes the neighbors of $u$ in this fixed round.
	With probability at least $(1-2^{i}/d)^{d/2^i +1} \geq 1/4$, we find at least one marked not.
	Now, node $w$ has degree at most $d/2^{i-1}$.
	Thus, the probability that no neighbor of $w$ (which is not in $N(u)\cup \{u\}$) is marked is at least $(1-2^{i}/d)^{d/2^{i-1}}\geq 1/16$.
	Hence, overall, with probability at least $1/4 \cdot 1/16 = 1/64$, there is a marked node $w \in N(u)\cup \{u\}$ for which no neighbor is marked.
	That is, with at least a constant probability $c>0$, node $u$ or one of its neighbors joins the independent set $S$.
	This means that node $u$ is effectively removed (i.e., it will be removed, the first time that it wakes up).
	Moreover, since the randomness of the $C(\log d - i)^2$ rounds of this phase are independent, the probability that node $u$ remains with a degree at least $d/{2^i}$ at the end of this phase (i.e., does not get effectively removed by having an independent set node in its inclusive neighborhood) is at most $2^{-5(\log d-i+1)^2}$.
	This holds by choosing $C$ large enough, depending on the constant probability $c=1/64$ above.
	Moreover, we can also union bound over all such nodes $u$ in the $r_i=(\log d-i)$-hop neighborhood of $v$.
	Since $\mathcal{E}_{i-1}$ did not happen, the degree of each node in this neighborhood is at most $(d/2^{i-1})$, which means the number of nodes in this neighborhood is at most $(d/2^{i-1})^{\log d-i} = 2^{(\log d- i +1) (\log d-i)}$.
	Hence, we have
	\begin{align*}
		\Pr[\mathcal{E}_i]
		 & \leq \Pr[\mathcal{E}_{i-1}] + 2^{(\log d- i +1) (\log d-i)} \cdot 2^{-5(\log d-i+1)^2} \\
		 & \leq 2^{-2(\log d-i+2)^2} + 2^{(\log d- i +1) (\log d-i)} \cdot 2^{-5(\log d-i+1)^2}   \\
		 & \leq 2^{-2(\log d-i+2)^2} + 2^{-5(\log d-i+1)^2}                                       \\
		 & \leq 2^{-2(\log d-i+1)^2}/2+2^{-2(\log d-i+1)^2}/2                                     \\
		 & =    2^{-2(\log d-i+1)^2}.
	\end{align*}

	The above completes the induction proof.
	Now, by applying the claim to $i=\log d$, we get that the probability of node $v$ having a degree at least $1$ by the end of phase $\log d$ is at most $1/2$.
	Hence, with probability at least $1/2$, node $v$ is isolated by the end of the $\log d$ phases and thus gets removed at the latest in the last round of the iteration.

	Finally, since different iterations use independent randomness, we get that the probability of a node $v$ staying after $k$ iterations is at most $1/2^k$.
	Hence, in expectation, at most $n/2^{k}$ nodes remain after iteration $k$.
	Moreover, we can apply extensions of the Chernoff bound for low-dependency random variables (cf. \Cref{thm:Chernoff-Extension}) to prove that a similar upper bound holds with high probability: Notice that the events of two nodes $v$ and $u$ that have distance greater than $O(\log4\log n)$ from each other are independent, as the algorithm, even including all the $O(\log\log n)$ iterations, has round complexity $O(\log^4\log n)$.
	Hence, for each node $v$, except for at most $\Delta'=(O(\log^2 n))^{O(\log^4\log n)} = 2^{O(\log\log^5 n)}$ nodes that are within distance $O(\log^4\log n)$ hops of $v$, the event of $v$ remaining is independent of all other nodes.
	That is, these events of different nodes remaining have dependency degree at most $\Delta'$.
	Hence, using \Cref{thm:Chernoff-Extension}, we can conclude that for any $k\in [1, O(\log\log n)]$, the probability of more than $O(n/2^{k})$ nodes remaining after $k$ iterations is at most
	\begin{align*}
		O(\Delta') \cdot exp\big(-\Theta(\frac{n}{2^k\Delta'})\big)
		 & \leq exp\big(O(\log^5\log n) -\frac{n}{2^{O(\log^5\log n)}}\big) \\
		 & \ll exp(-n^{0.9}) \ll 1/\poly(n).
	\end{align*}
	Thus, w.h.p., at most $O(n/2^{k})$ nodes remain after the $k^{th}$ iteration.
	This is the property that we used above when calculating the node-averaged awake complexity and proving property (B) of the lemma.

	Moreover, by setting $k=O(\log\log n)$ iterations, we get that at most $O(n/\log n)$ nodes remain after the last iteration, with high probability, which proves property (C) of the lemma.
\end{proof}

\subsection{Part III} By the end of part II of the algorithm, at most $O(n/\log n)$ nodes remain in the graph, i.e., nodes $v$ such that $v$ does not have a node in the  $S$ in its inclusive neighborhood $N(v) \cup \{v\}$.
In Part III of the algorithm, which is a simple clean-up process, we run a standard MIS algorithm--- e.g., Luby's algorithm~\cite{luby86}--- on the remaining nodes.
This algorithm runs in $O(\log n)$ rounds, with high probability.
Moreover, it adds extra nodes to the independent set $S$ and ensures that $S$ is a maximal independent set.
Finally, since we run this algorithm on only the remaining nodes, which are at most $O(n/\log n)$ many, these $O(\log n)$ rounds contribute just an additive $O(1)$ rounds to the overall node-averaged awake complexity.

\subsection{Wrap Up}
Finally, we show how to combine the previous free parts, and provide a formal proof of \Cref{thm:main}.
\begin{proof}[Proof of \Cref{thm:main}]
	As explained in \Cref{lem:part1}, the first part runs in $O(\log n)$ rounds.
	Moreover, only $O(n/\log n)$ nodes are awake during this time, which implies an $O(n)$ contribution to the summation $\sum_{v} A_v$.

	Also, as explained in \Cref{lem:part2}, the second part runs in $O((\log\log n)^2)$ rounds, and moreover, the average awake complexity among the nodes that participate in this part is $O(1)$.
	Hence, this part also contributes  $O(n)$ contribution to the summation $\sum_{v} A_v$.
	Finally, as shown in \Cref{lem:part2}, at most $O(n/\log n)$ nodes remain after the second part.

	The third part is simply running Luby's $O(\log n)$-round algorithm~\cite{luby86} on all the remaining nodes.
	Since at most $O(n/\log n)$ nodes remain for the third part, this part also contributes at most $O(n)$ to the summation $\sum_{v} A_v$.

	Overall, considering all the three parts, we have an algorithm with round complexity $O(\log n) + O((\log\log n)^4) + O(\log n) = O(\log n)$.
	Moreover, the total awake complexity is $\sum_{v} A_v = O(n) + O(n) + O(n) = O(n)$.
	Hence, the node-averaged awake complexity is $\frac{1}{n}\sum_{v} A_v = O(1)$.
\end{proof}

\section{Matching Approximation}
In this section, we first give an algorithm that computes a $(2 + \epsilon)$-approximation of maximum matching, and show that it also allows us to obtain an algorithm for getting a $(2 + \epsilon)$-approximation of vertex cover.
Finally, we show that the approximation ratio of the matching algorithm can be improved to $(1 + \epsilon)$, with only a constant overhead in the round and average awake complexity, for any constant $\epsilon$.

\subsection{Vanilla Algorithm for Approximate Fractional Matching}

We start with the following simple algorithm for computing a fractional matching.
This algorithm or close variants of it have been used throughout prior work~\cite{kuhn16_jacm,ghaffari2018improved,fischer2020improved} (though, the result has often been mentioned as a known procedure and we are unsure where the first appearance was).
This algorithm by itself has an $\Theta(\log \Delta)$ average awake complexity, as it requires nodes to be active at all times, and is therefore insufficient for our purpose of having an $O(1)$ average awake complexity.
We will later discuss how we use careful random sampling and estimation ideas to develop an alternative algorithm based on the same framework that achieves an $O(1)$ average awake complexity.

\paragraph{Vanilla Algorithm}
Each edge $e$ starts with a value of $x_e = \frac{1}{\Delta}$.
Initially, all edges are \emph{active}.
Then, in each following round, any node $v$ for which we have $c_v = \sum_{e \ni v} x_e \geq 1 - \epsilon$ is called \emph{tight}, and we \emph{freeze} all edges incident to tight nodes.
Finally, we update the value $x_e$ of all edges that are still active (i.e., not frozen) to $(1 + \epsilon) x_e$, and proceed to the next round.

\paragraph{Analysis}
Notice that this process finishes in $O(\log_{1+\eps} \Delta)$ rounds, as the maximum value an edge can get is 1.
For the quality of the computed solution, we have the following guarantee.
\begin{claim}
	Let $M^*$ be a maximum matching in $G$, and for each edge $e$, let $x_e$ be its value computed by the above algorithm.
	Then, we have $\Abs{M^*} \leq (2+4\epsilon) \sum_{e \in E} x_e$, i.e. the sum of all fractional values is a $(2 + 4 \epsilon)$-approximation for the maximum (fractional) matching, and the values of $x_e$ form a matching, i.e. for each vertex $v$ we have $c_v = \sum_{e \ni v} x_e \leq 1$.
	\label[claim]{clm:vanilla}
\end{claim}
\begin{proof}
	First, let us argue that the computed values $x_e$ form a valid fractional matching.
	In the first step, we set $x_e = \frac{1}{\Delta}$, so the condition $c_v \leq 1$ is trivially fulfilled for all vertices.
	For the following iterations, we only increase the weight of an edge if both its endpoints satisfy $c_v \leq 1- \epsilon$, thus multiplying the weight of all edges that are not frozen with a factor of $1 + \epsilon$ will still guarantee $c_v \leq (1 + \epsilon) \cdot (1 - \epsilon) = 1 - \epsilon^2 \leq 1$ for all vertices $v$.
	This concludes the first part.

	For why the fractional matching is a $(2 + 4 \epsilon)$-approximation, note that in the end, no edge is still active, so every edge is incident to at least one vertex $v$ for which $c_v > 1 - \epsilon$.
	We can think of a charging process where we start with values on the maximum (fractional) matching $M^*$ and we move these locally to our computed fractional matching so that each edge $e$ receives a charge of at most $x_e \cdot 2 (1 + 2 \epsilon)$, hence proving that $\Abs{M^*} \leq (2+4\epsilon) \sum_{e \in E} x_e$.
	Let us  start with the maximum (fractional) matching $M^*$.
	We move all fractional value of an edge $e = (u,v)$ in $M^*$ to its tight endpoint, call it $v$.
	Since $M^*$ is a valid matching, the sum of all fractional values that are moved to any node $v$ is at most 1, which is at most $\frac{1}{1-\epsilon} c_v$.
	We then spread these weights to the fractional matching that we have computed, and thus each edge $e$ receives at most $x_e \cdot \frac{1}{1-\eps}$ from each of its endpoints, which means at most $x_e \cdot \frac{2}{1-\eps} \leq x_e \cdot 2 (1 + 2 \epsilon)$ in total.
	Formally, this observation can be written as
	\begin{align*}
		\Abs{M^*} \leq \frac{1}{1 - \epsilon} \cdot \sum_{v \in V} c_v = 2 (1 + 2 \epsilon) \cdot \sum_{e \in E} x_e.
	\end{align*}
\end{proof}

\paragraph{Randomized Rounding} Given a fractional matching, it is well-known that one can use a simple randomized rounding method to obtain an integral matching, with the same size up to a constant factor (in expectation, or in fact with a concentration exponential in the matching size).
We make use of one such known method by Ghaffari et al.~\cite{ghaffari2018improved}.
To be self-contained, we provide a sketch of the rounding procedure and discuss the expected integral matching size, but we refer to ~\cite{ghaffari2018improved} for the probability concentration argument.

\begin{lemma}[Ghaffari et al~\cite{ghaffari2018improved} ]
	\label[lemma]{lem:rounding}
	There is an $O(1)$ round algorithm that, given a fractional matching $M$ with size $\Abs{M} = \sum_{e\in E} x_e$, computes an integral matching with size at least $|M|/50$, with probability $1-exp(-\Theta(|M|))$.
\end{lemma}
\begin{proof} [Proof Sketch of \Cref{lem:rounding}]
	Each node $v$ either proposes to one of its neighbors, choosing neighbor $u$ with probability $x_{e}/10$ where $e=\{v, u\}$, or does not propose to any neighbor, which happens with probability $1-\sum_{e \ni v} x_e$.
	An edge $e=\{v, u\}$ is marked if $v$ proposed to $u$ or $u$ proposed to $v$.
	We keep marked edges that are not incident on any other marked edge.
	This process can be clearly implemented in $O(1)$ rounds of message passing.
	It is also easy to see that an edge $e$ is marked with probability at least $x_{e}/10$, and conditioned on that, there is a probability of $1-(2\cdot 2/10) = 0.6$ that no other edge incident on $e$ is marked.
	This already shows that the expected size of the matching is at least $3|M|/50$.
	Ghaffari et al.~\cite{ghaffari2018improved} show, via an application of McDiarmid's inequality, that there is indeed an exponential concentration around this expectation and in particular the probability that the matching size is below $|M|/50$ is at most $exp(-\Theta(|M|))$.
\end{proof}

\subsection{\texorpdfstring{Improving Average Awake Complexity to $O(1)$}{Improving Average Awake Complexity to O(1)}}
\label{sec:ImpAvgComp}
The above vanilla algorithm requires all vertices that are not yet frozen to be active in all rounds.
Thus, we do not directly get an algorithm with low average awake-complexity.
To improve the awake complexity to $O(1)$, we modify the algorithm such that in each iteration, instead of considering all neighbors to calculate $c_v$ exactly, the algorithm will only consider a small subset of vertices as awake, and each vertex will use the decisions of its awake neighbors to estimate $c_v$ and make its decisions accordingly.

Let us look at the first few rounds of an initial attempt towards an algorithm to highlight the main ideas.
Initially, all edge weights are set to $\frac{1}{\Delta}$ as in the vanilla algorithm, so the decision of whether or not to freeze a node is almost the same as estimating its degree.
Thus, if we sample nodes with probability $p_1 = 10\log n\frac{1}{\Delta}$, each node $v$ can distinguish w.h.p. if it's degree is greater than $(1 - \epsilon) \Delta$ or smaller than $(1 - 20 \epsilon) \Delta$, e.g. by comparing the number of sampled neighbors to $p_1 \cdot (1- 10 \epsilon) \Delta$.

For the second round, we have to already consider two cases to estimate the value of a node $v$:
First, there are neighbors of $v$ that were frozen in the first round, an whose edge incident to $v$ has value $\frac{1}{\Delta}$ and there are neighbors that remained active and have weight $\frac{1 + \epsilon}{\Delta}$ on their incident edge.
To remedy this, we sample twice, once to estimate the number of neighbors of each kind.
Let $S_1$ be the set obtained by sampling nodes with probability $p_1 = 10\log n\frac{1}{\Delta}$, and let $S_2$ be the set obtained by sampling with probability $p_2 = 10 \log n \frac{(1 + \epsilon)}{\Delta}$.
We can now estimate the value of a node after two rounds as
\[
	\tilde{c}_v = \frac{\Abs{S_1 \cap N(v)}}{p_1} \cdot \frac{1}{\Delta} + \frac{\Abs{S_2 \cap N(v)}}{p_2} \cdot \frac{\epsilon}{\Delta}.
\]
This corresponds to the contribution of nodes that are active in the first round, and the contribution of the nodes that remain active into the second round.
After the first round, the maximum degree of active nodes can only be roughly $\frac{\Delta}{1 + \epsilon}$.
Thus, the sampling probability $p_2$ is chosen inversely proportional to the degree, with an added factor of $10 \log n$ to guarantee that the estimation is accurate w.h.p.

However, this approach has one crucial issue, namely that it still does not provide us with the desired awake complexity.
Once we proceed to further rounds, we will have $p_i = 10 \log n \frac{(1 + \epsilon)^i}{\Delta}$, which summed over all $i$ still yields an awake complexity of $O(\log n)$.\footnote{With slightly more care, this idea can actually be turned into an algorithm with awake complexity $O(\log \log n)$.}
The only way to get around this problem will be to reduce the sampling probability (particularly in later rounds), which will require us to analyze what happens if nodes make different decisions than they would in the vanilla algorithm.

For the algorithm, we will proceed in phases, where we set a sampling probability for each phase.
In phase 1, this will be $p_1 = \frac{C}{(\log n)^4}$.
This probability guarantees that as long as the degree of active nodes is $(\log n)^5$, the algorithm will behave the same way as the vanilla algorithm, w.h.p.
For the second phase, we will set $p_2 = \frac{C}{(\log \log n)^4}$, however, this will then also decrease our probability of correctly estimating the values of nodes to $\text{poly}(\frac{1}{\log n} )$, as long as the degree is at most $(\log \log n)^5$.
One final technicality is that instead of considering the degree of active nodes, we use the inverse of the weight of the active edges as an approximation for the degree.

\paragraph{Full Algorithm}
More generally, the algorithm is as follows, from the viewpoint of a node $v$:
We start in round $j = 0$ with all nodes being active.
Let $w_j = \frac{(1 + \epsilon)^{j}}{\Delta}$ be the weight of active edges in round $j$.
We say that a round $j$ is in phase $i = i(j)$ if $i$ is the smallest value such that $w_j \leq \frac{1}{(\log^{(i)} n)^{5}}$.
Further, let $p_i = \frac{C}{(\log^{(i)} n)^{4}}$.
In round $j$, we sample a set $S_h^j$ for each previous round $h \leq j$, by including each node independently with probability $p_{i(h)}$.
Then, let $A_h^j \subseteq S_h^j$ be the set of nodes that are also active.
We remark that even for the same value of $h$, but different values of $j$ we sample different sets, to preserve independence between rounds.
Then, we have that $\Abs{A_h^j \cap N(v)}$ is the number of neighbors of $v$ that were active in round $h$ and sampled (with probability $p_{i(h)}$).
For round $j$, a node $v$ now estimates the value $\tilde{c}_v$ as
\begin{align*}
	\tilde{c}_v = \sum_{h = 1}^j \frac{\Abs{A_h \cap N(v)}}{p_{i(h)}} \cdot (w_h - w_{h-1}).
\end{align*}
Then, if $\tilde{c}_v > (1 - 10 \epsilon)$, we freeze vertex $v$, i.e. it will not be active in future iterations.
If not, we continue to the next iteration.

We stop this part of the algorithm when we reach a phase $i$, such that $\frac{1}{\log^{(i)} n} > \epsilon \cdot \frac{\log(1 + \epsilon)}{1000}$.
At this stage, we switch to executing the vanilla algorithm, i.e. from now on we set the sampling probability to 1 for all future rounds.

Once the algorithm has finished, i.e. when all edges are frozen, we set the final edge weight as follows:
For edge $e = (u, v)$, let $j$ be the last round for which $u$ and $v$ were both still active.
Then we set $x_e = \frac{(1 + \epsilon)^j}{\Delta}$.
Finally, each vertex computes its final value $c_v = \sum_{e \ni v} x_e$, and we only keep vertices for which $c_v \leq 1$.
That is, for all vertices where $c_v > 1$, we set the weight of all incident edges to $0$.

\paragraph{Distributed Implementation}
We think of the algorithm in the following way:
In each distributed round $k$, each node $v$ that decides to be active should be able to simulate all rounds $j \leq k$ of the above algorithm.
Thus, it needs to know all sets $A_h^j$ for all $h \leq j$ and $j \leq k$.
This can be achieved recursively:
All nodes $u \in S_h^j$ will need to know whether or not they were active in round $h$, thus they need to know if they were frozen in round $h-1$ (or before).
So, for each $h' < h$, they need to know all sets $A_{h'}^j$ to make said decision, which we do recursively.

So the distributed implementation is as follows (formulating the recursion in a bottom up way):
Initially, we sample all sets $S_h^j$, such that each node $v$ knows all values of $j$ and $h$ such that $v \in S_h^j$.
In round 0, all nodes that are in $S_0^{j}$ for any $j$ become active and since they are all active at the beginning of round $0$, they all decide they are in $A_0^j$ too, which they send to all their neighbors, in this round and all future rounds, together with all values of $j$ for which this is case.
Also in round 0, all nodes that are in $S_1^j$ for any $j$ listen to which of their neighbors are in $A_0^1$, and use this information to decide if they are frozen in round 0 of the algorithm, i.e. if they are in $A_1^j$ too.
From round 1 on, the nodes that are in any $A_1^j$ then send this information to all their neighbors together with the value of $j$.
More generally, in round $h-1$, all nodes $v$ for which there exists (at least one) $j$ such that $v \in S_h^j$ become active and receive the information about $N(v) \cap A_l^{h'}$ for $l \leq h'$ and $h' < h$.
Then, node $v$ uses the sets $A_l^{h'}$ to simulate round $h'$ of the algorithm, i.e. to decide if it was frozen in round $h'$ for all $h' < h$, and it stops this simulation if it is ever frozen.
If $v$ is never frozen, then it is active in round $h$, and thus $v$ is in $A_h^j$ for all $j$ such that $v \in S_h^j$.

To summarize, for a node $v$ let $h$ be the smallest value such $v \in S_h^j$.
Then, $v$ is active from round $h-1$ until the end of the algorithm.
Finally, at some point all nodes will be active, from which on we simply execute the vanilla algorithm in a distributed fashion.

\subsection{Roadmap}
Before giving a formal analysis, we first give an outline and explanation of our proof strategy.
In particular, we will compare our algorithm to the vanilla algorithm and see how they differ.

There are two kinds of errors that can happen in the estimation of the value of a node $v$;
(1) the value $c_v$ is under-estimated, which means a node could proceed to the next round, even though increasing the weight of the incident edges could increase the value of $v$ beyond 1.
(2) $c_v$ is over-estimated, i.e. even though $c_v$ has not reached value $1 - \Theta(\epsilon)$ yet, the node $v$ could mistakenly be frozen.
These errors affect the quality of the final solution in different ways: while (1) destroys the fact that the output is a valid fractional matching, (2) makes it hard to argue about the quality of the computed solution.

Let us first focus on issue (1), and let's call nodes that have $c_v > 1$ \emph{heavy}.
In the algorithm we already deal with this problem by just removing all heavy vertices, however, for the analysis, we also need to bound how much value is lost this way.
Note that simply bounding the number of heavy nodes is not enough, as $c_v$ could grow to be $\Theta(\Delta)$.
Thus, we look at the value incident to those heavy nodes and mark the value on each of these incident edges as \emph{spoiled}.
Looking at edges instead of nodes allows us to more easily bound the additional amount of fractional value that is spoiled in each round, as edges can have value at most 1.
This is formalized in \Cref{lem:HeavyBadRound} and used to show that only a very small amount of the current fractional assignment is spoiled in every round.

For the second problem (2) let us call the nodes that are frozen with $c_v < 1  - \Theta(\epsilon)$ \emph{light}.
It is fairly straightforward to show that only a small fraction of all nodes are \emph{light} with high probability, however, each of those nodes could have ``blocked'' up to $\Theta(1)$ total fractional value.
If the size of the maximum fractional matching is e.g. just $\sqrt{n}$, this means that the amount of fractional value lost due to \emph{light} nodes would be much larger than the size of the computed fractional assignment itself.
To remedy this, we will look more carefully at the probability that a \emph{light} node with value $c_v$ is frozen too early, in particular we will show that this probability is proportional to $c_v$ (see \Cref{lem:LightBadNode}).
By doing so we can show that the expected number of light nodes that are frozen in any given round is only a small fraction of the size of the current fractional assignment (\Cref{lem:LightBadRound}), which we show can not get much larger than the maximum fractional matching.

Finally, we will put everything together in \Cref{thm:MatchingApprox} to show that a similar charging argument as in the proof of \Cref{clm:vanilla} can be used to show that the size of the fractional matching computed by the algorithm is a $(2 + \epsilon)$-approximation of the optimum fractional matching, in expectation.
For graphs where the maximum matching is sufficiently large, this also holds with high probability.

\subsection{Analysis of the the Matching Algorithm}
We start with some preliminaries:
We will assume that $\epsilon$ is sufficiently small, i.e. $\epsilon \leq 1/1000$, since it will simplify some calculations.
Further, we will use the following definition:
\begin{definition}
	Consider a node $v$ that is still active in round $j$ of the algorithm.
	Let each of its incident edges $e = (v, u)$ have value $x_e = w_k$ if $u$ was frozen in round $k < j$ and value $x_e = w_j$ if $u$ is also active in round $j$.
	We will say that $v$ has (true) \emph{value} $c_v = \sum_{e \ni v} x_e$ at the beginning of round $j$.
\end{definition}
Note that this is the same as stopping the algorithm before round $j$ and computing the values as at the end of the algorithm.
Using this definition, we can now show that the probability of a node becoming \emph{heavy} or being heavy and not frozen is small.
\begin{lemma}
	In round $j$ of phase $i$, the probability that a node $v$ with value $c_v > 1 - \epsilon$ estimates its value as $\tilde{c}_v < 1 - 10 \epsilon$ is at most $\frac{1}{1000 (\log^{(i-1)} n)^{10}}$.
	\label[lemma]{lem:HeavyBadNode}
\end{lemma}

We can now use this statement to bound the amount of fractional value that will be newly marked as \emph{spoiled} in one round of the algorithm.
Recall that we said the fractional value of an edge $e$ is spoiled, if $e$ is incident to a heavy node.
Thus, to bound the amount of value that is additionally spoiled in round $j$, we just need to look at all nodes that become heavy or stay heavy in a round $j$.
\begin{lemma}
	In every round $j$ in phase $i$, the extra fractional value that will be added by nodes that have value $c_v > 1 - \epsilon$, i.e. the amount that is additionally spoiled, is at most $\frac{100}{(\log^{(i)} n)^2} \cdot \max \{(\log n)^{18}, \Abs{M_j} \}$ with probability at least $1 - n^{-9}$, where $M_j$ is the fractional assignment at the beginning of round $j$.
	If the round is in phase $i > 1$, we assume that the degree among the remaining active vertices is at most $(\log n)^5$.
	\label[lemma]{lem:HeavyBadRound}
\end{lemma}
Before we can now use this to give a statement about the total amount of fractional that is spoiled in the execution of the algorithm, we need to ensure that the condition of \Cref{lem:HeavyBadRound} holds, which we achieve with the following Lemma.
\begin{lemma}
	After phase $i = 1$, the degree in the graph induced by all active nodes is at most $(\log n)^5$ with probability $1 - n^{-8}$.
	\label[lemma]{lem:LowDegree}
\end{lemma}
Finally, we are now able to bound the amount of fractional value that is spoiled in the entire algorithm.
\begin{lemma}
	The total fractional value that is spoiled, i.e. that is incident to vertices with value $c_v > 1$, is at most $\epsilon \max \{ (\log n)^{18}, \Abs{M} \}$ with probability $1 - n^{-7}$, where $M$ is the fractional assignment computed by the algorithm before the final removal of heavy nodes.
	\label[lemma]{lem:HeavyBad}
\end{lemma}
We now turn our attention to vertices that are \emph{light}, i.e. that were frozen too early by the algorithm.
Our first result says that vertices are unlikely to be light, in fact we even prove a stronger statement than for the heavy case.
We show that the probability that a node with a small value $c_v$ is frozen too early is proportional to $c_v$, i.e., the smaller the value, the less likely the node is to be frozen.
\begin{lemma}
	The probability that a vertex $v$ of value $c_v < 1 - 20 \epsilon$ is frozen in round $j$ in phase $i$ is at most $c_v \cdot \frac{1}{10 (\log^{(i-1)} n)^{10}}$.
	\label[lemma]{lem:LightBadNode}
\end{lemma}
\begin{proof}
	Let $X$ be the random variable denoting the value of $\tilde{c}_v$ in round $j$.
	Since $\tilde{c}_v$ is an unbiased estimator of $c_v$, we have $\E[X] = c_v$.
	We are interested in bounding $\Pr[X > 1 - 10 \epsilon]$, which is exactly the probability that node $v$ is (mistakenly) frozen.
	First, note that in the same way as in \Cref{lem:HeavyBadNode} we have that the probability that
	\begin{align*}
		\Pr[X > 1 - 10\epsilon] & \leq \Pr[ \Abs{X - \E[X]} \geq 9 \epsilon] \\
		                        & \leq \frac{1}{1000(\log^{(i-1)} n)^{10}}
	\end{align*}
	Thus, for $c_v \geq 1/100$ the statement is true, and we will assume $c_v \leq 1/100$ in the following.

	Note that we can write $X$ as follows:
	\begin{align*}
		X = \sum_{h \leq j} \frac{w_h}{p_{i(h)}} \sum_{\substack{u \in N(v) \\ u \text{ active in round $h$}}} X_{u, h},
	\end{align*}
	where $X_{u, h}$ is the indicator variable for the event that node $u$ is sampled in round $h$.
	We can rewrite the probability that a note is prematurely frozen as:
	\begin{align*}
		\Pr [ X > 1 - 10 \epsilon] = \Pr [X > ( 1 + \frac{1 - 10 \epsilon - c_v}{c_v} ) c_v ].
	\end{align*}
	In the following, let $\delta = \frac{1 - 10 \epsilon - c_v}{c_v}$.
	Note that since $\epsilon \leq 1/1000$ and $c_v \leq 1/100$, we have that $\delta \geq 980$ and trivially $\delta \leq \frac{1}{c_v}$.

	In order to get a stronger bound from \Cref{thm:ChernoffTight} we aim to maximize the expected value of the analyzed random variable $X$ while ensuring that each of the coefficients of the random binary variables is less than $1$.
	We have that the coefficients are of the form $\frac{w_h}{p_{i(h)}}$ where $h \leq j$.
	This is at most $\frac{1}{C (\log^{(i)} n)}$ as argued previously.
	Thus we can apply \Cref{thm:ChernoffTight} to $C (\log^{(i)} n) X$ and obtain
	\begin{align*}
		 & \Pr [ C (\log^{(i)} n)  X > (1 + \delta ) \E[C (\log^{(i)} n)^5 X] ]                         \\
		 & \leq \left( \frac{e^{\delta}}{(1 + \delta)^{(1 + \delta)}} \right)^{ C (\log^{(i)} n) \E[X]} \\
		 & \leq \exp \left( (\delta - (1 + \delta) \ln(1 + \delta) ) C (\log^{(i)} n) c_v \right)       \\
		 & \leq \exp \left( (1 - \ln(1 + \delta)) \delta  C (\log^{(i)} n) c_v \right)                  \\
		 & \leq \exp \left( (- \ln(1 + \delta)/2 - 1) (1/c_v)  C (\log^{(i)} n) c_v \right)             \\
		 & \leq c_v \cdot \frac{1}{10 (\log^{(i-1)} n)^{10}}.
	\end{align*}
\end{proof}
In the next Lemma, we will make use of the fact that the probability that a light node is frozen is proportional to its value to show that the total number of light nodes that are frozen is proportional to the size of the current matching.
\begin{lemma}
	The number of vertices $v$ that are frozen in round $j$ of phase $i$ despite having $c_v < 1 - 20 \epsilon$, is at most $\frac{100}{(\log^{(i)} n)^2} \cdot \max \{ (\log n)^{18},  \Abs{M_j} \}$ with probability at least $1 - n^{-9}$, where $M_j$ is the value of the fractional assignment at the beginning of round $j$.
	If the round $j$ is in phase $i > 1$, we assume that the degree among the remaining active vertices is at most $(\log n)^5$.
	\label[lemma]{lem:LightBadRound}
\end{lemma}
\begin{proof}
	Let $X_v$ be the indicator random variable that is one if node $v$ estimates its value $\tilde{c}_v > 1 - 10 \epsilon$ despite having $c_v < 1 - 20 \epsilon$.
	From \Cref{lem:LightBadNode} we know that $\E[X_v] \leq c_v \cdot \frac{1}{10 (\log^{(i-1)} n)^{10}}$.
	Further, as $\sum_{v} c_v = 2 \Abs{M_j}$ we have that for $X = \sum_v X_v$ it holds $\E[X] \leq 2 \Abs{M_j} \frac{1}{10 (\log^{(i-1)} n)^{10}}$.
	By definition, $X$ is the number of nodes that are prematurely frozen in round $j$.

	Let us first look at rounds in phase $i = 1$.
	There we have that each individual node $v$ is prematurely frozen with probability at most $\frac{1}{10 n^{10}}$, which implies by a union bound that any of the $n$ nodes in the graph are prematurely frozen is at most $n^{-9}$.

	For nodes in round $i > 1$, we have to be more careful with our analysis, as the failure probabilities of individual nodes are too small for us to just union bound over all nodes.
	Instead, we will aim to use \Cref{thm:Chernoff-Extension} and bound the dependency degree.
	Note that the event of a node being prematurely frozen is only dependent on which neighbors are sampled
	Thus two nodes that do not share any neighbors are independent.
	By our assumption we have that the degree among active vertices is at most $(\log n)^5$, meaning the dependency degree is $\Delta' \leq (\log n)^{10}$.

	For the final step, we consider two cases:
	(1) $\Abs{M_j} \leq (\log n)^{18}$ and (2) $\Abs{M_j} > (\log n)^{18}$.
	In the first case, we will analyze $\Pr[X \geq \frac{100}{(\log^{(i)} n)^2} \cdot (\log n)^{18}]$, while we know that $\E[X] \leq \frac{1}{5 (\log^{(i-1)} n)^{10}} (\log n)^{18}$.
	Thus, we have
	\begin{align*}
		 & \Pr[X \geq \frac{100}{(\log^{(i)} n)^2} \cdot (\log n)^{18}]                                                         \\
		 & \leq \Pr[ \Abs{X - \E[X]} \geq \frac{99}{(\log^{(i)} n)^2} (\log n)^{18} ]                                           \\
		 & \leq 8 ((\log n)^{10} + 1) \cdot \exp \left(- \frac{99 (\log n)^{18}}{3((\log n)^{10} + 1) (\log^{(i)} n)^2 }\right) \\
		 & \leq 8 ((\log n)^{10} + 1) \cdot \exp \left( - 10 (\log n)^2 \right) \leq n^{-9},
	\end{align*}
	For case (2) where $\Abs{M_j} > (\log n)^{18}$, we have the same argument, except that instead of $(\log n)^{18}$ we have $\Abs{M_j}$, and we obtain (at worst) the same probability of failure since we have $\Abs{M_j} > (\log n)^{18}$.
\end{proof}
Finally, we are now able to bound the total number of light nodes at the end of the algorithm.
\begin{lemma}
	The number of vertices that are frozen too early, i.e. that have value $c_v < 1 - 20 \epsilon$, is at most $\epsilon \max \{(\log n)^{18}, \Abs{M}\}$, where $M$ is the fractional assignment computed by the algorithm, with probability $1 - n^{-7}$.
	\label[lemma]{lem:LightBad}
\end{lemma}
\begin{proof}
	Our goal is to repeatedly apply \Cref{lem:LightBadRound}, however, we have one condition after phase 1 that needs to be satisfied.
	We can use \Cref{lem:LowDegree} which says that this condition is satisfied with probability at least $1 - n^{-8}$, allowing us to assume that the degree among active vertices is at most $(\log n)^5$.

	We call a round successful if it satisfies the statement of \Cref{lem:LightBadRound}.
	Since we have at most $\log_{1 + \epsilon} n \ll n$ rounds overall the probability that all rounds are successful is at least $1 - n^{-8}$.

	In order to know how often we need to apply \Cref{lem:LightBadRound}, we will also need to know the number of rounds in a phase $i$.
	By the condition that a round $j$ is in phase $i$ such that $i$ is that smallest value for which $w_j \leq \frac{1}{(\log^{(i)} n)^3}$, and $w_j = \frac{(1+\epsilon)^j}{\Delta}$ we have that phase $i$ consists of less than $\frac{5 \log^{(i)} n}{\log (1 + \epsilon)}$ rounds.

	This allows us to bound the total number of nodes that are prematurely frozen by
	\begin{align*}
		 & \sum_{j = 1}^{\log_{1 + \epsilon} \Delta} \frac{100}{(\log^{(i(j))} n)^2} \cdot \max \{ (\log n)^{18},  \Abs{M_j} \}                 \\
		 & \leq \sum_{i} \frac{5 \log^{(i)} n}{\log (1 + \epsilon)} \cdot \frac{100}{(\log^{(i)} n)^2} \cdot \max \{ (\log n)^{18},  \Abs{M} \} \\
		 & = \frac{500}{\log (1 + \epsilon)} \cdot \max \{ (\log n)^{18},  \Abs{M} \} \sum_{i} \frac{1}{(\log^{(i)} n)}.
	\end{align*}
	As we stop the algorithm once $\frac{1}{\log^{(i)} n} > \epsilon \cdot \frac{\log(1 + \epsilon)}{1000}$ we can again bound this by
	\begin{align*}
		\epsilon \cdot \max \{ (\log n)^{18},  \Abs{M} \}.
	\end{align*}
	Overall, the probability of failure is at most $n^{-8} + n^{-8} \leq n^{-7}$, proving our statement.
\end{proof}
We are now able to prove the main result of this section.
Combining \Cref{lem:HeavyBad} and \Cref{lem:LightBad} allows us to show that the charging argument from \Cref{clm:vanilla} can still be used, while only losing a small amount of the fractional value due to the nodes that behaved different than in the vanilla algorithm.
We remark that we did not attempt to optimize the constants, especially the exponent of the lower bound on the matching size.
\begin{theorem}
	Let $\bar{M}$ be the matching computed by the algorithm.
	With probability $1 - n^{6}$ this matching $\bar{M}$ is a $(2+100\epsilon)$-approximation of a maximum (fractional) matching $M^*$, i.e. $(2 + 100\epsilon) \Abs{\bar{M}} \geq \Abs{M^*}$, in graphs where $\Abs{M^*} \geq \Omega((\log n)^{18})$.
	\label[theorem]{thm:MatchingApprox}
\end{theorem}
\begin{proof}
	First, note that $\bar{M}$ is indeed a valid matching, as each node has value $c_v \leq 1$, independent of the execution of the algorithm.
	Let $M$ be the fractional assignment computed by the algorithm before the deletion of all nodes $v$ with value $c_v > 1$.
	First, we will give an upper bound on the value of $M$ in terms of $M^*$.
	\Cref{lem:HeavyBad} says that the fractional value incident to nodes which have $c_v > 1$ is at most $\Abs{M} - \Abs{\bar{M}} \leq \epsilon \max\{ (\log n)^{18}, \Abs{M} \}$, with probability $1 - n^{-7}$.
	Since $\bar{M}$ is a valid matching $\Abs{\bar{M}} \leq \Abs{M^*}$.
	Combining this we get
	\[
		\Abs{M} - \epsilon \max\{ (\log n)^{18}, \Abs{M} \} \leq \Abs{M^*},
	\]
	from which follows, using $\Abs{M^*} \geq (\log n)^{18}$,
	\[
		\Abs{M} \leq (1 + 2\epsilon) \Abs{M^*}.
	\]

	With this observation, we can now also get a lower bound on $\Abs{M}$ by using \Cref{lem:LightBad}, which holds with probability $1 - n^{-7}$.
	Out of all vertices that the algorithm froze, all but $\epsilon \max \{(\log n)^{18}, \Abs{M}\}$ have value at least $1 - 20 \epsilon$.
	In order to use a similar argument as for the vanilla algorithm, let us think that we add fractional value 1 to all vertices that were prematurely frozen, which means we add at most a value of $\epsilon \max \{(\log n)^{18}, \Abs{M}\}$ overall.
	Now, every edge is incident to at least one endpoint with a value of at least $1 - 20 \epsilon$.
	Thus, we have
	\[
		\Abs{M^*} \leq \frac{2}{1 - 20 \epsilon} ( \Abs{M} + \epsilon \max \{(\log n)^{18}, \Abs{M}\} ).
	\]
	For $\Abs{M^*} = \Omega( (\log n)^{18})$ large enough, this implies $\Abs{M} \geq (\log n)^{18}$, which allows us to simplify
	\[
		\Abs{M^*} \leq \frac{2(1 + \epsilon)}{1 - 20 \epsilon} \Abs{M} \leq (2 + 50\epsilon) \Abs{M}.
	\]
	Finally, as previously observed, we have $\Abs{M} \leq \frac{1}{1-\epsilon} \Abs{\bar{M}}$.
	Combining this with the above we get
	\[
		\Abs{M^*} \leq \frac{(2 + 50\epsilon)}{1 - \epsilon} \Abs{M} \leq (2 + 100 \epsilon) \Abs{\bar{M}}.
	\]
	The failure probability is $n^{-7} + n^{-7} \leq n^{-6}$.
\end{proof}
\begin{remark}
	Using a similar line of argumentation, it can also be shown that a fractional matching which is a $(2 + 100 \epsilon)$-approximation \emph{in expectation} can also be computed, with no lower bound on the maximum matching size.
\end{remark}
Having shown that the algorithm computes a good approximation of the maximum matching, we now also show that it is implementable in the sleeping model with low awake complexity.
\begin{corollary}
	We can compute a $(2 + 100 \epsilon)$-approximation of maximum fractional matching in any graph $G$ with maximum fractional matching size at least $\Omega((\log n)^{18})$, in the distributed sleeping model with awake complexity $O(\epsilon^{-3})$ with probability $1 - n^{-5}$ and round complexity $O(\log \Delta)$ (always).
	\label{cor:awakecomp}
\end{corollary}
\begin{proof}
	Note that the approximation guarantee holds from \Cref{thm:MatchingApprox} with probability $1 - n^{-6}$, so we only need to know the second part of the statement, i.e. the complexity in the sleeping model.
	For the round complexity, we have that the algorithm trivially stops when edge weights $w_j$ are 1, so the round complexity is bounded by $\log_{1+\epsilon} \Delta = O(\log \Delta)$.
	The awake complexity is analyzed in two parts, first, we have the awake complexity that occurs due to nodes being awake in phases with sampling probability $p_i < 1$ and then we also have to bound the number of rounds where we just simulate the vanilla algorithm.

	For the first part, we have that for each set $S_h^{j}$ where $h \leq j$ the nodes in $S_h^j$ become active in round $h-1$ and stay active for the remainder of the algorithm.
	If $h$ is in phase $i$, the set $S_h^j$ has expected size $p_i \cdot n$, and we can use \Cref{thm:hoeffding} to conclude that with probability at least $n^{-10}$, we have $\Abs{S_h^j} \leq 2 \cdot p_i \cdot n$.
	Further, with probability at least $n^{-9}$ this holds for all rounds and all sets $S_h^j$.
	Note that from the start of phase $i$ to the end of the algorithm, there are at most $\frac{5 \log^{(i)} n}{\log (1 + \epsilon)}$ rounds.
	Thus, for a fixed value of $h$ in phase $i = i(h)$, we have that there are at most $\frac{5 \log^{(i)} n}{\log (1 + \epsilon)}$ many sets $S_h^j$.
	For each of these sets, all nodes are active for at most $\frac{10 \log^{(i)} n}{\log (1 + \epsilon)}$ rounds.
	And finally note that there are also at most $\frac{5 \log^{(i)} n}{\log (1 + \epsilon)}$ rounds in phase $i$.
	Putting all this together, we get a total awake complexity of
	\begin{align*}
		\sum_i \frac{250 (\log^{(i)} n)^3}{(\log (1 + \epsilon))^3} \cdot \frac{2 C \cdot n}{(\log^{(i)} n)^4}
		 & = \frac{500 \cdot C \cdot n}{(\log (1 + \epsilon))^3} \sum_i \frac{1}{\log^{(i)} n}                                                    \\
		 & \leq \frac{500 \cdot C \cdot n}{(\log (1 + \epsilon))^3} \cdot 2 \cdot \epsilon \cdot \frac{\log(1+\epsilon)}{1000} = O(n/\epsilon^3).
	\end{align*}
	Thus, the average awake complexity due to that part is $O(\epsilon^{-3})$.

	For the second part, we stop once we would start a phase $i$, where $\frac{1}{\log^{(i)} n} > \epsilon \cdot \frac{\log(1 + \epsilon)}{1000}$.
	As before, note that from the start of phase $i$ until we have that weights $w_j \geq 1$, there are at most $\frac{5 \log^{(i)} n}{\log (1 + \epsilon)}$ rounds.
	We can rewrite this as
	\[
		\frac{5000}{\epsilon (\log (1 + \epsilon))^2} = O(\epsilon^{-3}).
	\]
	Thus, the average awake complexity of the entire algorithm is $O(\epsilon^{-3}) + O(\epsilon^{-3}) = O(\epsilon^{-3})$, and the overall probability of failure is $n^{-6} + n^{-9} \leq n^{-5}$.
\end{proof}

\subsection{Vertex Cover}

In this section we consider the vertex cover problem, which is the dual of the matching problem discussed previously.
We will first show that we can also obtain a $(2 + \epsilon)$-approximation to the minimum vertex cover using the vanilla algorithm, and then give a corresponding proof for the algorithm with low awake complexity.
\begin{claim}
	Let $C^*$ be a vertex cover of minimum size in $G$, and let $C$ be the set of all vertices that are frozen in the vanilla algorithm.
	Then, $C$ is a valid vertex cover, and we have $\Abs{C} \leq (2 + 4 \epsilon) \Abs{C^*}$.
\end{claim}
\begin{proof}
	First, let us argue that $C$ is indeed a vertex cover.
	As every edge has at least one frozen endpoint at the end of the algorithm, the set of all frozen vertices is a vertex cover.
	Now for the approximation guarantee we use a similar charging argument as in the matching case.
	Let us place one dollar on each vertex $v$ in $C$.
	Now, distribute this dollar to all incident edges, proportional to their fractional value $x_e$ at the end of the algorithm.
	Thus, each such edge receives at most $\frac{1}{1 - \epsilon} x_e$ dollars from $v$, as $v$ was a vertex with $c_v \geq 1 - \epsilon$.
	Since an edge can receive this amount from each endpoint, overall an edge receives at most $\frac{2}{1 - \epsilon} x_e$ dollars.
	Now we move the value from all edges to their endpoint that is in $C^*$, breaking ties arbitrarily.
	Since we have that each vertex had $c_v \leq 1$, each vertex from $C^*$ receives at most $\frac{2}{1 - \epsilon}$ dollars this way.
	Thus, we have $\Abs{C} \leq (2 + 4 \epsilon) \Abs{C^*}$.
\end{proof}

\begin{theorem}
	\label[theorem]{thm:VertexCover}
	Let $C$ be the set of all vertices that are frozen by the algorithm in \Cref{sec:ImpAvgComp}.
	Then, $C$ is a vertex cover, and with probability $1 - n^{-6}$ it is a $(2 + 100 \epsilon)$-approximation of the minimum vertex cover $C^*$, in graphs where $C^* \geq \Omega((\log n)^{18})$.
\end{theorem}
\begin{proof}
	First, note that $C$ is a valid vertex cover, as every edge is incident to at least one vertex that was frozen.
	Since $C^*$ is a minimum vertex cover, we have that $\Abs{C^*} \geq \Abs{M^*}$ and $\Abs{C^*} \leq 2 \Abs{M^*}$ for a maximum matching $M^*$.
	Note that this also implies that $\Abs{M^*} \geq \Omega((\log n)^{18})$.
	Let $M$ be the fractional assignment computed by the algorithm, before the removal of all heavy nodes $v$ with $c_v > 1$ in the final step.
	As in the proof of \Cref{thm:MatchingApprox} we have that $\Abs{M} \leq (1 + 2 \epsilon) \Abs{M^*} \leq (1 + 2 \epsilon) \Abs{C^*}$ with probability $1 - n^{-7}$.
	Using \Cref{lem:LightBad}, we get that the number of nodes that finish with value $c_v < 1 - 20 \epsilon$ is at most $\epsilon \Abs{M} \leq 2 \epsilon \Abs{C^*}$ with the same probability.
	Further, we also have that the total amount of fractional value incident to nodes that have $c_v > 1$ is at most $\epsilon \Abs{M} \leq 2 \epsilon \Abs{C^*}$ by \Cref{lem:HeavyBad} with probability $1 - n^{-7}$ too.
	Thus, the overall failure probability is $3 n^{-7} \leq n^{-6}$.

	We can now do the same kind of charging argument as in the vanilla case.
	Let us place 1 dollar on every vertex in $C$.
	We remove this one dollar from all the light vertices with $c_v \leq 1 - 20 \epsilon$.
	As argued previously, this amounts to at most $2 \epsilon \Abs{C^*}$ many dollars being removed.
	Then, we distribute all remaining dollars proportional to the weight of the incident edges $x_e$.
	As we have that $c_v \geq 1 - 20 \epsilon$ for all vertices that do this, each edge receives value at most $\frac{2}{1 - 20 \epsilon} x_e$, since it can receive this value from both endpoints.
	Now, we would like to redistribute this again to the vertices that are in $C^*$, however, we might have nodes that have $c_v \gg 1$.
	As in the beginning, we remove all fractional value incident on these nodes with $c_v > 1$.
	This means we remove at most $\frac{2}{1 - 20 \epsilon} 2 \epsilon \Abs{C^*}$ dollars.
	With those vertices removed, we can now move the remaining dollars to vertices of $C^*$ such that every vertex in $C^*$ obtains at most $\frac{2}{1 - 20 \epsilon}$ dollars.
	Overall, this yields the following inequality:
	\begin{align*}
		\Abs{C} - 2 \epsilon \Abs{C^*} - \frac{4}{1 - 20 \epsilon} \epsilon \Abs{C^*} & \leq \frac{2}{1 - 20 \epsilon} \Abs{C^*},
	\end{align*}
	which implies
	\begin{align*}
		\Abs{C}  \leq \frac{2 + 6 \epsilon}{1 - 20\epsilon} \Abs{C^*} \leq (2 + 100 \epsilon) \Abs{C^*}.
	\end{align*}
\end{proof}
In the same way as for the matching, the algorithm can be implemented in the distributed sleeping model, which gives us the following:
\begin{corollary}
	We can compute a $(2 + 100 \epsilon)$-approximation of minimum vertex cover in any graph $G$ with minimum vertex cover size $\Omega((\log n)^{18})$, in the distributed sleeping model with awake complexity $O(\epsilon^{-3})$ with probability $1 - n^{-5}$ and a deterministic round complexity of $O(\log \Delta)$.
\end{corollary}

\subsection{Improved Approximation of Maximum Matching}

In this section, we present a generic and black-box way of transforming any algorithm that computes a constant approximation for the maximum matching problem into an algorithm that computes a $(1 + \epsilon)$-approximation, via a constant number of invocations of the former, for any constant $\epsilon>0$.
In terms of the method, our reduction is inspired by the classic work of Goldberg et al.\cite{goldberg1993sublinear} which provided a sublinear-time parallel algorithm for computing the maximum matching.

There is also an algorithm due to McGregor~\cite{mcgregor2005finding}, which provides a somewhat similar reduction from $(1+\epsilon)$-approximation of maximum matching to the \emph{maximal matching} problem (this was primarily presented for the semi-streaming model, but it can be adapted to the distributed models of computing). This reduction is randomized (even in bipartite graphs) and requires $(1/\eps)^{O(1/\eps)}$ invocations. However, this reduction is not directly sufficient for our setting, as we do not have an algorithm for maximal matching with $O(1)$ average awake complexity. We believe that the reduction of McGregor~\cite{mcgregor2005finding} can be extended to produce an alternative reduction from $(1 + \epsilon)$-approximation to constant approximation of maximum matching. Furthermore, subsequent to the conference version of our paper, such a reduction was also given by Fischer, Mitrovi{\'c}, and Uitto~\cite{fischer2021deterministicarxiv} in an updated version of their published work~\cite{fischer2022deterministic}. Their reduction has a better $\epsilon$-dependency of $\poly(1 / \epsilon)$ and it also works deterministically in general graphs.

However, we find our reduction significantly simpler than both of the above results. Our reduction is deterministic in bipartite graphs and requires only $\poly(1/\eps)$ invocations of the constant approximation subroutine. In general graphs, where it becomes randomized, it involves $2^{O(1/\eps)}$ invocations. 

Before we state the algorithm in its full generality, we will first make two simplifying assumptions:
The black-box algorithm computes a maximal matching, and the input graph is bipartite.
We will later discuss how to remove these assumptions.
\begin{definition}
	Given a graph $G$ and a matching $M$ a path $P = (v_0, v_1, \dots, v_{2\ell+1})$ is an \emph{augmenting} path with respect to $M$, if the following hold:
	for each $1 \leq i \leq \ell$ the edge $(v_{2i-1}, v_{2i}) \in M$ and both $v_0$ and $v_{2 \ell + 1}$ are free vertices, i.e. there are no vertices $u, w$ such that $(v_0, u) \in M$ or $(v_{2 \ell + 1}, w) \in M$.
\end{definition}
Augmenting paths are one of the standard tools in algorithms for maximum matching, this dates back to Berge's theorem\cite{berge1957two}, and it has been used extensively since the seminal work of Hopcroft and Karp~\cite{HopcroftKarp1973}.
The path has its name due to the fact that it is possible to use it to increase the size of $M$, or augment $M$.
Namely, if the edges on $P$ are $e_1, e_2, \dots, e_{2 \ell + 1}$, with $e_{2i} \in M$ for $1 \leq i \leq \ell$, then we can remove all edges $e_{2i}$ from $M$ and add all edges $e_{2j+1}$ for $0 \leq j \leq \ell$ to obtain a matching $M'$ with $\Abs{M'} = \Abs{M} + 1$.

Another property of augmenting paths used in \cite{HopcroftKarp1973} is the following:
If we consider a graph $G$ and a matching $M$ such that there are no augmenting paths of length at most $2\ell - 1$.
Let $S$ be a set of disjoint augmenting paths of length $2 \ell + 1$, such that $S$ is maximal, i.e. no augmenting path of length $2 \ell + 1 $ can be added to it without violating the disjointedness of all pairs of paths.
Then, augmenting along all paths in $S$ gives a matching for which no augmenting path of length at most $2 \ell + 1$ exists.
Thus, the minimum length of any augmenting path increased, which will be our notion of progress.
We will then combine it will the following result, allowing us to bound the quality of the computed matching.
\begin{lemma}[Hopcroft, Karp \cite{HopcroftKarp1973}]
	Let $M^*$ be a maximum matching in $G$ and $M$ be a matching.
	If we have that $\Abs{M^*} > (1 + \epsilon) \Abs{M}$ there is an augmenting path of length at most $4 \epsilon^{-1} + 1$.
	\label[lemma]{lem:AugPathApprox}
\end{lemma}
\begin{proof}
	Let us look at the graph $G' = (V(G), M^* \oplus M')$, which is the symmetric difference of $M^*$ and $M$.
	Note that $G'$ contains isolated vertices, cycles of even length, and paths of odd and even length.
	Further, note that any path of odd length must be an augmenting path, whereas all other structures contain the same number of edges for $M$ and $M^*$.
	Also note that augmenting along all these paths would turn $M$ into a matching of the same size as $M^*$, and augmenting along each such path increases the size of $M$ by one.
	Thus there have to be at least $\Abs{M^*} - \Abs{M} > \epsilon/2 \Abs{M^*}$ many such augmenting paths.
	Suppose for contradiction that all augmenting paths are longer than $4 \epsilon^{-1} + 1$.
	Then, each augmenting path would contain at least $2 \epsilon^{-1}$ many edges from $\Abs{M}$.
	As we have at least $\epsilon/2 \Abs{M^*}$ such paths, this would imply that $\Abs{M} > \Abs{M^*}$, a contradiction.
	Thus an augmenting path of length $4 \epsilon^{-1} + 1$ must exist.
\end{proof}

\paragraph{Bipartite Graphs}
We will first state the algorithm for bipartite graphs.
The general outline is to start with a matching and improving it step by step.
In a first step, we will try to make it maximal, i.e. augmenting along all paths of length 1.
For the subsequent steps, we will go from a graph where there is no augmenting path of length at most $2i-1$ to a graph where there is no augmenting path of length at most $2i+1$.
We first assume that we have access to an algorithm that computes a maximal matching, since it simplifies part of the description and argument, while preserving the main ideas.
\begin{lemma}
	Suppose there is an algorithm $\mathcal{A}$, which computes a maximal matching in a bipartite graph.
	Suppose that we are also given a bipartite graph $H = (L \cup R, E)$, and a matching $M$ in $H$ such that there is no augmenting path of length $\leq 2i - 1$ where $i\leq O(1/\epsilon)$.
	Then, by invoking $\mathcal{A}$ at most $\Theta(\epsilon^{-3})$ times, we can compute an induced subgraph $H'$ of $H$, together with a set $S$ of augmenting paths of length $2i+1$ which is maximal in $H'$ and such that $\Abs{V(H)} - \Abs{V(H')} \leq \epsilon^2 \Abs{M}$.
	\label[lemma]{lem:maximalImprovedStep}
\end{lemma}
\begin{proof}
	We will proceed with the proof in 3 steps:
	First, we will show how to find a subgraph of $H$ that captures all relevant information, then we discuss how the algorithm works on this subgraph, and finally we give an analysis for why the algorithm performs as stated.

	\paragraph{Layer graph}
	First, we create a \emph{layer graph} as follows:
	Layer $L_0$ contains all unmatched vertices in $L$.
	Layer $L_1$ contains all neighbors of vertices in $L_1$.
	Then, layer $L_2$ contains all vertices that are connected by an edge from $M$ to a vertex in $L_1$,
	Then, we create layer $L_3$ in the same way as $L_1$, only considering vertices that are not already contained in some layer, and so on.
	We stop at layer $L_{2i+1}$, in which we only include vertices from $R$ that are not incident to any edge in $M$.

	An alternative way of defining the same graph is as follows: First, we orient all edges that are not in $M$ from $L$ to $R$ and all edges in $M$ from $R$ to $L$.
	Then, we perform a breadth-first search (BFS) according to this orientation, for $2i+1$ steps, and starting from all unmatched vertices in $L$ simultaneously.
	We then discard any edge that are not between consecutive layers, any vertices that were not added to any layer during these $2i+1$ steps, and we discard also vertices that were added to $L_{2i+1}$ but are matched.

	The rest of the algorithm is performed only on this layer graph.
	Note that any augmenting path $P = (v_0, \dots, v_{2i+1})$ of length $2i+1$ is contained in this layer structure such that $v_j \in L_j$.
	This is since on one hand $P$ must be contained in the layer graph by construction, and if we assume that any $v_j \notin L_j$, then we also have that $v_{2i+1}$ is in a layer before $L_{2i+1}$ which would mean there exists an augmenting path of length at most $2i-1$.
	Thus we know that we only have to consider augmenting paths that respect the structure of the layer graph.

	\paragraph{Augmentation algorithm on the layer graph}
	The augmentation algorithm work as follows:
	We start with a set of paths $S_0$ which is all paths of length zero starting at the vertices of $L_0$.
	In the first step, we compute a matching between $S_0$ and all their neighbor vertices, which is precisely $L_1$.
	We expand the paths of all vertices that were matched by two hops, first to their matched neighbor and then along the edge in $M$ which connects this neighbor to $L_2$.
	All unmatched vertices are removed from further consideration in the algorithm.
	For future steps, the procedure is almost the same:
	in each step $j$ we try to expand the paths of $S_j$ by computing a matching between their endpoints and all their outgoing neighbors.
	However, if a path of length $>1$ is not matched, we do not delete it, but instead we ``backtrack'' and deactive its last two vertices, meaning they won't participate for the rest of this algorithm.
	Then, we add it to the sets of paths $S_{j+1}$ for the next iteration.
	All paths that are matched are then expanded as in the first step and added to $S_{j+1}$.
	If a path reaches $L_{2i+1}$, i.e. the last layer, it will stop participating in the algorithm and say that it becomes \emph{inactive}.
	Overall, we execute this process for $\Theta(\epsilon^{-3})$ iterations.
	Finally, we return the set of all paths that reached $L_{2i+1}$ as $S$, and remove all vertices on paths that are still active, to get the subgraph $H'$.
	However, all vertices that were deactivated during the algorithm as part of the backtracking, are still part of $H'$.

	\paragraph{Analysis} We claim that after $\Theta(\epsilon^{-3})$ iterations, at most $\epsilon^3 \Abs{M}$ active paths remain.
	We will show that there are at most $6 \Abs{M}$ updates that can occur (after the first iteration) and each active path causes at least one update per iteration.
	Note that in layers $L_1$ to $L_{2i}$ there are at most $2 \Abs{M}$ vertices, as all vertices of those layers are incident to edges from $M$.
	Further, each of these vertices can be updated at most twice, once when it is added to and once when it is potentially deleted from a path.
	The remaining $2 \Abs{M}$ are due to vertices from $L_0$, of which at most $\Abs{M}$ remain after the first iteration, and which can only be deleted once after, and from vertices of $L_{2i+1}$, which can only be updated when a path becomes inactive, which can happen at most once for each of the at most $\Abs{M}$ paths.
	Suppose we had more than $\epsilon^3 \Abs{M}$ active paths after $10 \epsilon^{-3}$ iterations.
	Each of those paths has caused at least one update for each iteration it was updated, thus overall there would be $10 \Abs{M}$ updates, which as we argued before is not possible.
	As every active path has at most $2i-1 \leq 1/\epsilon$ vertices, removing all active paths will remove at most $\epsilon^2 \Abs{M}$ vertices.

	Finally, we need to argue that $S$ is a maximal set of augmenting paths in $H'$.
	An alternative formulation of this maximality condition is to argue that removing all vertices on paths in $S$ will disconnect the first and last layers, i.e. $L_0$ and $L_{2i+1}$, in $H'$.
	We will prove this by showing that in each step of the algorithm, after the first one, we have two properties:
	(1) removing the paths from $S_j$ disconnects the layers $L_0$ and $L_{2i+1}$ in $H$, and
	(2) for all edges $e = (v_k, v_{k+1}) \in M$ that were deleted up to (and including) step $j$, there is no path from $v_{k+1}$ to $L_{2i+1}$ in $H \setminus S_{j}$.
	Let us start with the first step of the algorithm.
	Since we compute a maximal matching between $L_0$ and $L_1$, removing the endpoints of this matching will disconnect $L_0$ and $L_1$.
	The second property is trivially true, since there are no such edges yet.

	So suppose that properties (1) and (2) are satisfied after step $j$ of the algorithm, and we want to show that they hold after step $j+1$.
	We will start with property (2), which will then imply property (1).
	Consider an edge $e = (v_k, v_{k+1}) \in M$ that was removed in step $j+1$ of the algorithm.
	For each neighbor $u$ of $v_{k+1}$ in layer $L_{k+2}$, there are several reasons why $v_{k+1}$ was not matched to $u$.
	If $u$ is already part of an active path, or was matched in this round to another active path, we have that property (2) trivially holds.
	If $u$ (together with its incident edge from $M$) has been deleted in a previous iteration, then we can just use the inductive hypothesis and also immediately get property (2).
	So the remaining case is that $u$ was also deleted in iteration $j+1$ due to backtracking.
	However, in this case we can say we \emph{blame} the edge $(u, w) \in M$ for $v_{k+1}$ not being matched in step $j+1$.
	Note that $e$ can blame multiple edges.
	If a similar situation occurs for why $w$ was not extended, we also blame this new edge.
	Note that in each of theses blaming steps we advance one layer, and finally in layer $L_{2i+1}$ we cannot blame any more edges in $M$, since all vertices in $L_{2i+1}$ are unmatched.
	Thus, this procedure must end, or we would have an active path that could be extended to $L_{2i+1}$ contradicting the maximality of our computed matching.
	Now for property one, we have that for all paths that were extended, this trivially still holds, and for all other paths we have just shown property (2), so the set $S_{j+1}$ still disconnects $L_0$ and $L_{2i+1}$.
\end{proof}
We are now ready to extend the argument to the case when we only have an algorithm that computes a constant approximation of the maximum matching, and not a maximal one.
First, we will introduce an alternative notion of an approximation guarantee that an algorithm can have.
\begin{definition}
	We say that a matching $M$ is $\delta$-maximal in $G$, if the maximum matching size of $G \setminus V(M)$ is at most $\delta \cdot \Abs{M}$.
\end{definition}
In words, this means that removing the matching $M$ from $G$ gives a graph that does not contain a large matching.
We can also easily obtain such a $\delta$-maximal matching, given an algorithm for computing any constant approximation of maximum matching.
\begin{lemma}
	Given an algorithm $\mathcal{A}$ that computes a $c$-approximation of maximum matching, we can find a $\delta$-maximal matching $M$ using at most $O(1/c \log 1/\delta)$ invocations of $\mathcal{A}$.
	\label[lemma]{lem:deltaMax}
\end{lemma}
\begin{proof}
	We have $I=\Theta(1/c \log 1/\delta)$ iterations.
	In each iteration, we apply algorithm $\mathcal{A}$ on the subgraph induced by the remaining nodes, we add the matching that it computes to our output matching, and we remove all matched vertices from the graph, before proceeding to the next iteration.
	Let $|M^*|$ be the size of the maximum matching in the original graph.
	In the first iteration, algorithm $\mathcal{A}$ finds a matching with size $|M^*|/c$ and thus, the size of the maximum matching in the subgraph induced by the nodes that remain is at most $(1-1/c)|M^*|$.
	In the second iteration, algorithm $\mathcal{A}$ finds a matching whose size is at least a $1/c$ fraction of the maximum matching in the graph that remained after the first iteration.
	Hence, the size of the maximum matching in the graph that remains after the second iteration is at most $(1-1/c)^2|M^*|$.
	Similarly, and by simple induction, we see that after $i$ iterations, the size of the maximum matching in the graph that remains after the $i^{th}$ iteration is at most $(1-1/c)^i|M^*|$.
	Let $M$ be the output matching computed during $I$ iterations.
	We then know that $|M| + (1-1/c)^I|M^*| \geq |M^*|/2$, because the output matching along with a maximum matching of the remaining graph form a maximal matching of the entire graph and thus have size at least $|M^*|/2$.
	We conclude that $|M|\geq (1/2-(1-1/c)^I)|M^*|\geq |M^*|/3$.
	Therefore, $M$ is a $\delta$-maximal matching as $((1-1/c)^I \cdot |M^*|)/(|M^*|/3)=3(1-1/c)^I\leq 3 e^{-\Theta(\log 1/\delta)} \leq \delta$.
\end{proof}
Using this notion of $\delta$-maximality, we can now extend the result to algorithms that compute a $c$-approximation of maximum matching.
\begin{lemma}
	Suppose there is an algorithm $\mathcal{A}$, which computes a $c$-approximation of maximum approximation in bipartite graphs.
	Suppose we are also given a bipartite graph $H = (L \cup R, E)$, and a matching $M$ in $H$ such that there is no augmenting path of length $\leq 2i - 1$.
	Then, by invoking $\mathcal{A}$ at most $O(1 /c \cdot \epsilon^{-3} \cdot \log 1/\epsilon)$ times, we can compute a subgraph $H'$ of $H$, together with a set $S$ of augmenting paths of length $2i+1$ which is maximal in $H'$.
	Further, we have that $\Abs{V(H)} - \Abs{V(H')} \leq \epsilon^2 \Abs{M}$, and that the graph $H'' = (V(H'), E(H) \setminus E(H'))$ does not contain a matching of size $\epsilon^2 \Abs{M}$.
	\label[lemma]{lem:approxImproved}
\end{lemma}
\begin{proof}
	First, we will use \Cref{lem:deltaMax} and assume that instead of algorithm $\mathcal{A}$ we have an algorithm $\mathcal{A}'$, which computes a $\delta$-maximal matching.
	This only incurs a $O( 1/c \cdot \log 1/\delta)$ overhead in the number of calls to algorithm $\mathcal{A}$.
	The algorithm for finding the augmenting paths is the same as in \Cref{lem:maximalImprovedStep}, except that we compute a $\delta$-maximal matching, instead of a (fully) maximal one.
	After computing a matching, we remove all edges that are between unmatched vertices (note that in the case of a maximal matching, there would be no removed edges).
	At the end of the algorithm, we do the same as in \Cref{lem:maximalImprovedStep}, except that $H'$ is not an induced subgraph anymore, since we also removed some edges between unmatched nodes.

	Note that $H''$ contains all edges that were removed due to their endpoints being unmatched.
	In each step of the algorithm, the largest matching in this set of edges is at most $\delta \Abs{M}$, thus over all iterations the matching in $H''$ can have size at most $\Theta(\epsilon^{-3}) \delta \Abs{M}$, which is at most $\epsilon^2 \Abs{M}$ for $\delta = \Theta(\epsilon^5)$.

	We have that $\Abs{V(H)} - \Abs{V(H')} \leq \epsilon^2 \Abs{M}$ still holds the same way as in the case for maximal matching, so the only thing we need to argue is that the computed set $S$ is maximal in $H'$.
	Note that this is also still true, since all edges that could be added to extend the computed matching to a maximal one are removed, thus still keeping $L_0$ and $L_{2i+1}$ disconnected.

	We compute a $\delta$ maximal matching in each of the $\Theta(\epsilon^{-3})$ iterations, which each time incurs $O(1/c \cdot \log 1 / \epsilon)$ calls to $\mathcal{A}$.
	Thus, overall, the algorithm $\mathcal{A}$ is invoked $O(1 /c \cdot \epsilon^{-3} \cdot \log 1/\epsilon)$ times.
\end{proof}
Finally, we can now repeatedly use this to find augmenting paths of length $2i+1$, augment them, and proceed with a graph that does not contain any more augmenting paths of length at most $2i+1$.
We do this until we reach a stage where there are no more augmenting paths of length $5 \epsilon^{-1}$, at which point we know from \Cref{lem:AugPathApprox} that the computed matching must be a $(1 + \epsilon)$-approximation.
\begin{theorem}
	Suppose we are given an algorithm $\mathcal{A}$ which computes a $c$-approximation of maximum matching in bipartite graphs.
	Suppose we are also given a bipartite graph $H$.
	Then, we can find a $(1 + 3\epsilon)$-approximation of maximum matching in $H$ using at most $O(1/c \cdot \epsilon^{-4} \cdot \log 1/\epsilon)$ calls to $\mathcal{A}$.
	\label[theorem]{thm:BipMatchApprox}
\end{theorem}
\begin{proof}
	We start with an empty matching $M_0$ and use the algorithm from \Cref{lem:AugPathApprox} (which in turn uses $\mathcal{A}$) on $H_0 = H$ to find a set of augmenting paths of length 1 in a subgraph $H_1$.
	We augment along all these paths, i.e. add these edges to $M_0$ to create $M_1$.
	By the guarantees of \Cref{lem:AugPathApprox}, we get that after this augmentation there are no more augmenting paths of length 1 in $H_1$, so we continue to the next iteration.
	We invoke the algorithm from \Cref{lem:AugPathApprox} again, this time on the graph $H_1$, and setting $i = 1$.
	This yields another subgraph $H_2$ together with a set of augmenting paths.
	We again augment along all these paths, flipping their edges with respect to their membership in $M_1$, resulting in $M_2$.
	This leaves us with no more augmenting paths of length $3$ in $H_2$, and we continue like this for a total of $\ell = 2 \epsilon^{-1} + 1$ iterations, at which point we have a matching $M_\ell$ in a subgraph $H_\ell$ such that there are no augmenting paths of length $2\ell - 1 = 4 \epsilon + 1$ in $H_\ell$ with respect to $M_\ell$.
	By \Cref{lem:AugPathApprox}, we have that $M_\ell$ is a $(1 + \epsilon)$-approximation of maximum matching in $H_\ell$, but what can we say about the approximation quality of $M_\ell$ in $H$?
	We set $M = M_\ell$ and will show that $M$ is also a $(1 + 7\epsilon)$-approximation of maximum matching in $H$.

	In each step, we removed $\epsilon^2 \Abs{M}$ vertices, so over all $\ell = 2 \epsilon^{-1} + 1$ iterations we removed at most $3 \epsilon \Abs{M}$ vertices from $H$ to obtain $H_\ell$.
	However, we also removed some edges.
	In each iteration we removed a set of edges, which does not contain a matching of size larger than $\epsilon^2 \Abs{M}$.
	Thus, the matching in the union of all those edges can have size at most $\ell \cdot \epsilon^2 \Abs{M} \leq 3 \epsilon \Abs{M}$.
	Thus, the maximum matching $M^*$ in $H$ can have size at most $\Abs{M^*} \leq (1 + 7 \epsilon) \Abs{M}$.

	For the number of calls to algorithm $\mathcal{A}$, we have that each iteration yields $O(1/c \cdot \epsilon^{-3} \cdot \log 1/\epsilon)$ calls, so we have $O(1/c \cdot \epsilon^{-4} \cdot \log 1/\epsilon)$ over all iterations.
\end{proof}
\paragraph{General Graphs}
We will now show how to use the algorithm for bipartite graphs to also compute an improved matching in general graphs.
The main idea is to randomly bipartition the graph, and then try to augment the matching using just the edges that cross the bipartition.
\begin{theorem}
	\label[theorem]{thm:ImpMatchingApprox}
	Suppose we have an algorithm $\mathcal{A}$ that given a bipartite graph $H$ computes a $c$-approximation of the maximum matching in $H$.
	Then, for any graph $G$, we can find a matching that is a $(1 + \epsilon)$-approximation of the maximum matching in expectation, using at most $2^{O(1/\epsilon)}$ invocations of $\mathcal{A}$.
\end{theorem}
\begin{proof}
	To compute an initial matching $M_0$, we randomly partition the vertices of $G$ into two parts, delete all edges that are within the same partition, resulting in a graph $H$, and use $\mathcal{A}$ on $H$ to compute a $c$-approximation of the maximum matching in the resulting graph.
	Note that looking at a maximum matching $M^*$, each edge of $M^*$ is preserved in $H$ with probability $1/2$, so the expected size of the maximum matching in $H$ is at least $\Abs{M^*}/2$.
	Thus, in expectation $M_0$ is at least a $2c$-approximation of the maximum matching in $G$.

	Now we proceed in iterations $i = 1, 2, \dots$, where in iteration $i$ we start with a matching $M_i$.
	Suppose that $M_i$ is not a $(1 + \epsilon)$-approximation yet, i.e. we have that $\Abs{M^*} - \Abs{M_i} > \epsilon/2 \Abs{M^*}$.
	We randomly bipartition $G$ into two sets $L$ and $R$ and delete all edges that are within the same partition.
	Further, we also delete the endpoints of all edges in $M_i$ that are in the same partition.
	Let $H$ be the resulting graph.
	First, note that no matter how we change the matching in $H$, we can always add it back to $G$.
	That is, because we deleted all endpoints of matching edges that are in the same partition, thus those vertices cannot be matched in $H$.

	We will now argue that $H$ contains a sufficiently large potential for improvement:
	Since $M_i$ is not a $(1 + \epsilon)$-approximation, by a similar argument as in \Cref{lem:AugPathApprox} we have that there are at least $\epsilon / 4 \Abs{M^*}$ disjoint augmenting paths of length at most $\ell = 8 \epsilon^{-1} + 1$.
	Each of those paths is preserved with probability at least $2^{- \ell}$, i.e. all its edges are across the partition $L \cup R$.
	Thus, in expectation, we preserve at least $\epsilon / 4 \cdot 2^{- \ell} \Abs{M^*}$ many augmenting paths of length up to $\ell$.
	Let $M_H \subseteq M_i$ be the subset of $M_i$ that is preserved in $H$ and let $M_H^*$ be a maximum matching in $H$.
	Note that in expectation, $M_H^* \geq M^*/2$.
	Since for each preserved augmenting path, we can increase the size of $M_H$ by one, we have that (in expectation) $\Abs{M_H} + \epsilon / 4 \cdot 2^{- \ell} \Abs{M^*} \leq \Abs{M_H^*}$.
	We can now use \Cref{thm:BipMatchApprox} to compute a $(1 + \epsilon/8 \cdot 2^{-\ell} )$-approximation of $M_H^*$ in $H$, we call it $M_H'$.
	Comparing this to $M_H$ we get in expectation
	\begin{align*}
		\Abs{M_H'} - \Abs{M_H} & \geq (1 - \epsilon/8 \cdot 2^{-\ell} ) \Abs{M_H^*} - (1 - \epsilon/4 \cdot 2^{-\ell}) \Abs{M_H^*} \\
		                       & \geq \epsilon/8 \cdot 2^{-\ell} \Abs{M_H^*} \geq \epsilon/16 \cdot 2^{-\ell} \Abs{M^*}.
	\end{align*}
	This means that assuming the matching $M_i$ is not a $(1 + \epsilon)$-approximation, we can compute a matching $M_{i+1}$, that is by $\epsilon/16 \cdot 2^{-\ell} \Abs{M^*}$ larger than $M_i$.

	This can only happen at most $2^{O(1/\epsilon)}$ times, as after that we would have a matching of size larger than $M^*$.
	Thus, within $2^{O(1/\epsilon)}$ iterations we must reach a matching $M$, which is a $(1 + \epsilon)$-approximation to $M^*$ in expectation.
	We execute the algorithm $\mathcal{A}$ $2^{O(1/\epsilon)}$ time in each iteration, and there are $2^{O(1/\epsilon)}$ iterations, so the total number of calls to $\mathcal{A}$ is $2^{O(1/\epsilon)}$.
\end{proof}
\begin{remark}
	The augmentation also works if instead of a deterministic algorithm for computing a $c$-approximation of maximum matching we have an algorithm that outputs a $c$-approximation in expectation or with high probability.
\end{remark}

\subsection{Wrap Up}

\Cref{thm:mainMM} has two parts, one about matching, and about vertex cover.
For the vertex cover, the proof is given in \Cref{thm:VertexCover}.
For the matching, two more steps are necessary.
First, the algorithm only computes a fractional matching.
To get an integral matching, we can use \Cref{lem:rounding}, giving us a constant approximate integral matching (in expectation).
Finally, we can use this combination as a black-box in \Cref{thm:ImpMatchingApprox} to improve it to a $(1 + \epsilon)$-approximation.
For the node-averaged awake complexity, we have that it is constant for finding a constant approximation of maximum matching as shown in \Cref{cor:awakecomp}, the rounding takes a constant number of rounds in total, and then we only repeat it a constant number of times in the framework of \Cref{thm:ImpMatchingApprox}.

\section{Conclusion}

We showed that a maximal independent set can be computed with average awake complexity $O(1)$ while maintaining the round complexity of $O(\log n)$ of the standard Luby algorithm.
Further, we showed that the same bounds can be achieved to compute a $(1 + \epsilon)$-approximation of maximum matching as well as a $(2 + \epsilon)$-approximation of minimum vertex cover in expectation.

In future work, it would be interesting to investigate if one can obtain guarantees about the worst-case awake complexity that are better than the round complexity as well as if it is possible to compute a maximal matching with (any measure of) low awake complexity.

\section*{Acknowledgements}
The first author is grateful to Yuval Emek for discussions on an earlier attempt at the MIS problem.
The second author was supported by the Swiss National Foundation, under project number 200021\_184735.

\newpage
\bibliography{ref}
\bibliographystyle{alpha}

\end{document}